\documentclass[a4paper,11pt]{article}

\usepackage{jheppub} 

\usepackage{amssymb}
\usepackage{amsmath}
\usepackage{amsthm}
\usepackage{mathrsfs}
\usepackage[numbers]{natbib}
\usepackage{mathtools}
\usepackage{graphicx}
\usepackage{float}
\usepackage{subfigure}
\usepackage{booktabs}
\usepackage{multirow}
\usepackage[colorlinks=true]{hyperref}
\usepackage{physics}

\def\a{\alpha}
\def\b{\beta}
\def\sil{\tilde{S}_{I_l}}
\def\sjr{\tilde{S}_{J_r}}
\def\ts{\tilde{S}}
\def\sigx{\sigma(x)}

\def\mxx{m(x,x^\prime)}
\def\xx{(x,x^\prime)}
\def\fxx{(f(x),f(x^\prime))}
\def\ix{x^{(i)}}
\def\fix{f(x^{(i)})}
\def\p{\partial}
\def\cplm{\complement}

\def\tm{\tilde{m}}
\def\tS{\tilde{S}}
\def\Bp{\mathbb{B}^\prime}
\def\xp{x^\prime}

\newtheorem{theorem}{Theorem}
\newtheorem{definition}{Definition}
\newtheorem{lemma}{Lemma}

\title{\boldmath Improved proof-by-contraction method and  relative homologous entropy inequalities}

\author[a,b,c]{Nan Li,}
\author[a,b,c]{Chuan-Shi Dong,}
\author[d]{Dong-Hui Du}
\author[a,b,c,e,1]{and Fu-Wen Shu \note{Corresponding author.}}
\affiliation[a]{Department of Physics, Nanchang University, Nanchang, 330031, China}
\affiliation[b]{Center for Relativistic Astrophysics and High Energy Physics, Nanchang University, Nanchang, 330031, China}
\affiliation[c]{Institute for Advanced Study, Nanchang University, Nanchang, 330031, China}
\affiliation[d]{School of Physics and Astronomy, Sun Yat-sen University, Guangzhou 510275, China}
\affiliation[e]{Center for Gravitation and Cosmology, Yangzhou University, Yangzhou, China}
\emailAdd{nanli\_cn2000@163.com}
\emailAdd{dong\_chuanshi@163.com}
\emailAdd{donghuiduchn@gmail.com}
\emailAdd{shufuwen@ncu.edu.cn}

\abstract{The celebrated holographic entanglement entropy triggered investigations on the
connections between quantum information theory and quantum gravity. An important
achievement is that we have gained more insights into the quantum states. It allows us to
diagnose whether a given quantum state is a holographic state, a state whose bulk dual
admits semiclassical geometrical description. The effective tool of this kind of diagnosis is
holographic entropy cone (HEC), an entropy space bounded by holographic entropy
inequalities allowed by the theory. To fix the HEC and to prove a given holographic
entropy inequality, a proof-by-contraction technique has been developed. This method
heavily depends on a contraction map $f$, which is very difficult to construct especially for
more-region ($n\geq 4$) cases. In this work, we develop a general and effective rule to rule
out most of the cases such that $f$ can be obtained in a relatively simple way. In addition, we
extend the whole framework to relative homologous entropy, a generalization of holographic
entanglement entropy that is suitable for characterizing the entanglement of mixed states.}

\begin{document}
\maketitle
\flushbottom

\section{Introduction}
Recent progress in gravity/gauge duality has shed new light on our understanding of
quantum gravity. For instance, a significant insight into the connections between quantum
information theory and Anti-de Sitter/Conformal field theory (AdS/CFT) correspondence
\cite{AdS/CFT1, AdS/CFT2, AdS/CFT3}  was gained after the discovery of the holographic
entanglement entropy (HEE), which states that area of a codimension-one homologous surface
(the Ryu-Takayanagi (RT) surface) on a time slice is identified with the entanglement entropy in CFTs
\cite{HEE1, HEE2}
\begin{equation}\label{RTformula}
S(A)=\frac{\text{area}[m(A)]}{4G_{N}} ,
\end{equation}
where $m(A)$ is a spacelike minimal hypersurface in the bulk spacetime, anchored to the
asymptotic boundary such that $\partial m(A) = \partial A$. This remarkable formula
provides direct evidence of the potential connections between quantum information
theory and quantum gravity through holography.  As a consequence, over the past few
years a lot of efforts have been made to understand how this formula works in the
holography community.

One significant achievement of these efforts is that it allows us to gain further insights into
the quantum states. In particular, it provides an exact diagnosis on judging whether a given
quantum field theory state has a holographic duals \cite{Heemskerk:2009pn}. We usually
call those states whose bulk duals admit semiclassical geometric descriptions the holographic
states. These are a special sub-class of quantum states which satisfy certain constraints.
A significant fact is that the RT formula \eqref{RTformula} proved to be able to identify some
of these constraints in the form of inequalities that the entropies of boundary subregions
need to satisfy. In \cite{Bao:2015bfa}, the authors found that, by arranging subsystem
entropies of general mixed states into entropy vectors, entropy inequalities can be regarded
as a space bounded by the allowed entropy vectors. It turns out that this is a convex space
which has a configuration of a polyhedral cone and therefore is called the holographic
entropy cone (HEC).  HEC inequalities play a central role in understanding properties of
holographic states. Therefore, what is the general structure of HEC and how to prove HEC
inequalities become a key to the whole story. Over the last few years many progresses have
been made in this direction, see for example the incomplete list \cite{HEC1, HEC2, HEC3,
HEC4, HEC5, HEC6}.  Among them, a proof technique for holographic entropy inequalities
called proof-by-contraction was developed in \cite{Bao:2015bfa}.  It was latter generalized
to  more general cases in \cite{Rota:2017ubr,Avis:2021xnz}. Very recently, a geometric
formulation of this method has been proposed in \cite{Bao:2015bfa,Akers:2021lms}. Given a general
holographic entropy inequality
\begin{equation}
    \sum_{l=1}^{L}\alpha_lS_{I_l}\geq \sum_{r=1}^{R}\beta_rS_{J_r}
    \label{rtentropy}
\end{equation}
where $S_I$ is RT entropy, $\alpha_l,\beta_r$ are positive numbers, $L$ and $R$ are positive
integers, representing the numbers of entropy terms appearing on each sides.
The basic idea of  the proof-by-contraction is to construct a reasonable rule such that
one can compare terms in left-hand side (LHS) of \eqref{rtentropy} to the ones of right-hand
side (RHS) in full generality. It was found in \cite{Bao:2015bfa} that this can be achieved by
finding a contraction map $f$: $\{0, 1\}^L\rightarrow \{0, 1\}^R$, where
$\{0,1\}^L$ ( $\{0,1\}^R$) denotes a length-L (length-R) bitstrings formed by $0$ or $1$.
One can introduce a weighted Hamming distance $d_\omega$ which is defined as
\begin{equation}\label{domega}
    d_\omega(y,y^{\prime})\coloneqq\sum_{i=1}^m \omega_i\abs*{y_i-y^{\prime}_i},
\end{equation}
where $\omega$ is a weight vector, and $y, y' \in \{0, 1\}^m$ are a pair of bitstrings. We
can therefore prove that inequality \eqref{rtentropy} holds if there exists a contraction map
$f$ for the weighted Hamming distances $d_\alpha$ and $d_\beta$ satisfying
$d_\a(x,x^{\prime})\geq d_\b(f(x),f(x^{\prime}))$ for every $x, x' \in \{0, 1\}^L$
\cite{Akers:2021lms}.

One caveat, however, is that the difficulty of finding $f$ increases exponentially with the
number of the terms in \eqref{rtentropy}. If we use an exhaustive method
to find $f(x)$ for all $x$, the number of all possible situations is of order $(2^R)^{2^L}$ that is
tremendous. We need to develop a general and effective rule to rule out most of the
cases such that $f$ can be obtained in a relatively simple way. In this work, we develop
a set of general rules to get $f$ effectively. Several examples are given to show the
validity of our proposal.

In addition,  our method is valid not only for RT entropy, but also for relative homologous
entropy (RHE) \cite{Headrick:2017ucz}, a well-defined generalization of holographic
entanglement entropy that is suitable for characterizing the entanglement of mixed states.
As an important example which has been studied widely in the past few years, the
entanglement of purification (EoP) is a sub-class of RHE. See \cite{EoP,Takayanagi:2017knl,
Nguyen:2017yqw,EW1, EW2, EW3,HEoP1, HEoP2, HEoP3, HEoP4, HEoP5, HEoP6, HEoP7,
HEoP8, HEoP9, HEoP10, HEoP11, HEoP12, HEoP13, HEoP14, HEoP15, HEoP16, HEoP17,
HEoP18, HEoP19, HEoP20, HEoP21, HEoP22,HEoP23, HEoP24, HEoP25, HEoP26} for more
recent progress on EoP and its holographic duality. Our analysis is based on the
max flow-min cut (MFMC) theorem, which states that the maximum flux out of a boundary
region $A$, optimized overall divergence free norm-bounded vector fields in the bulk, is
proportional to the area of the minimal homologous surface $m(A)$  \cite{FH, BT2, BT3,
BT4, BT5,Agon:2021tia}.

This paper is organized as follows: In section 2, we give a brief review on MFMC theorem and
relative homology proposed in \cite{FH, Headrick:2017ucz}. In section 3, we first give an
introduction to the proof-by-contraction method, then we generalize this method to relative
homology and an improved version is developed. In section 4, we show how to devise a set
of general rules to find the contraction map $f$. In section 5, several examples are given to
show our method is valid.  A concluding remark is given in the last section.

\section{Brief review on MFMC theorem and relative homology}
In this section we would like to give a brief review on the Riemannian MFMC theorem and
relative homology in \cite{FH, Headrick:2017ucz}. First, let us begin with a compact oriented
manifold $M$ with a conformal boundary $\p M$. And define a ``flow" on $M$, which is a
vector field obeying
\begin{equation}
\nabla_\mu v^\mu=0,\qquad \abs{v}\leq 1.
\end{equation}
We set $4G_{N}=1$ for convenience. Considering a subregion $A$ on boundary $\p M$,
the MFMC theorem says that there exists a max flow $v^\mu$ satisfying \cite{FH}
\begin{equation}\label{mcmf}
    \sup_{v^\mu}\int_{A} \sqrt{h} n_\mu v^\mu=\inf_{m\sim A} \text{Area}(m) \equiv S(A),
\end{equation}
where $h$ is the determinant of the induced metric on $A$, $n_\mu$ is the normal vector.
And $m$ is a bulk surface homologous to $A$ with condition $\p A = \p m$
(denote $m$ is homologous to $A$ as $m\sim A$).
Where we have admitted the RT formula \cite{HEE1} in the second equal sign.

More generally, one can loosen the boundary condition on the bulk surface $m$
by introducing the concept of relative homology \cite{Headrick:2017ucz}.
Assume there is a surface $R$, which could be either on a boundary or in bulk.
And $m^{\prime}$ is the surface homologous to $A$, which is allowed to include part of $R$.
Then we can define a surface $m$ homologous to $A$ relative to $R$ as
\begin{equation}
m\coloneqq -(m^{\prime} \setminus R)
\end{equation}
excluding the part on $R$, which is denoted as $m\sim A \text{ rel } R$. Where the minus
sign means that $m$ takes different orientation relative to $m^{\prime}$.
See the Fig.~\ref{RH_simple}.

\begin{figure}[H]
    \centering
    \subfigure[Relative homology for one system]{
    \includegraphics[width=0.4\textwidth]{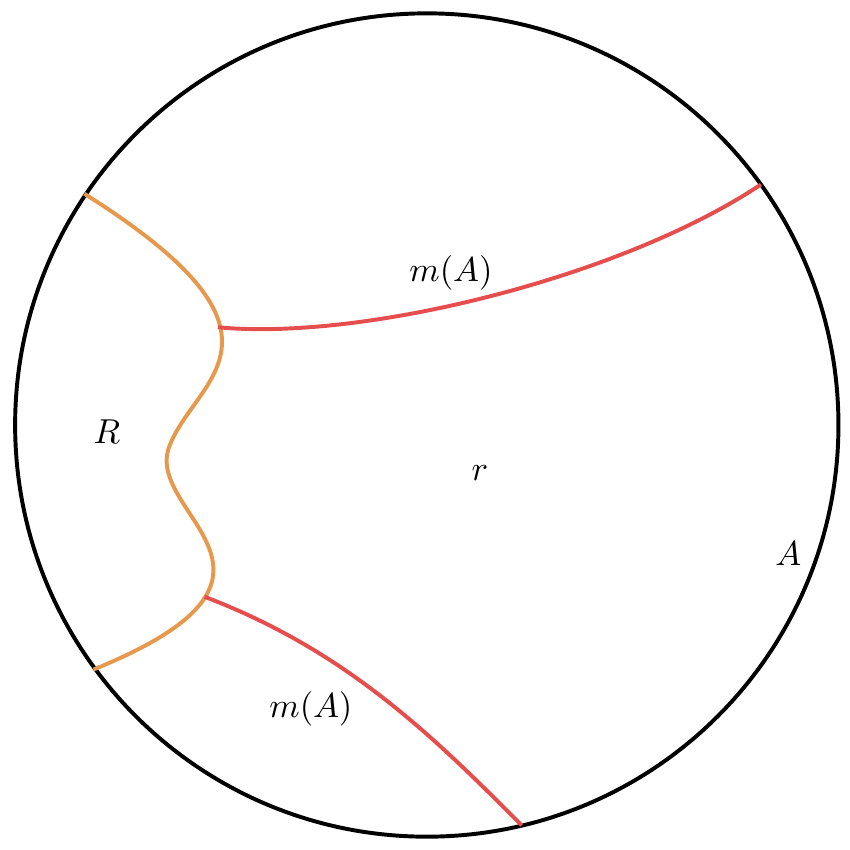}
    \label{RH_simple}}
    \subfigure[Relative homology for multiple systems]{
        \includegraphics[width=0.4\textwidth]{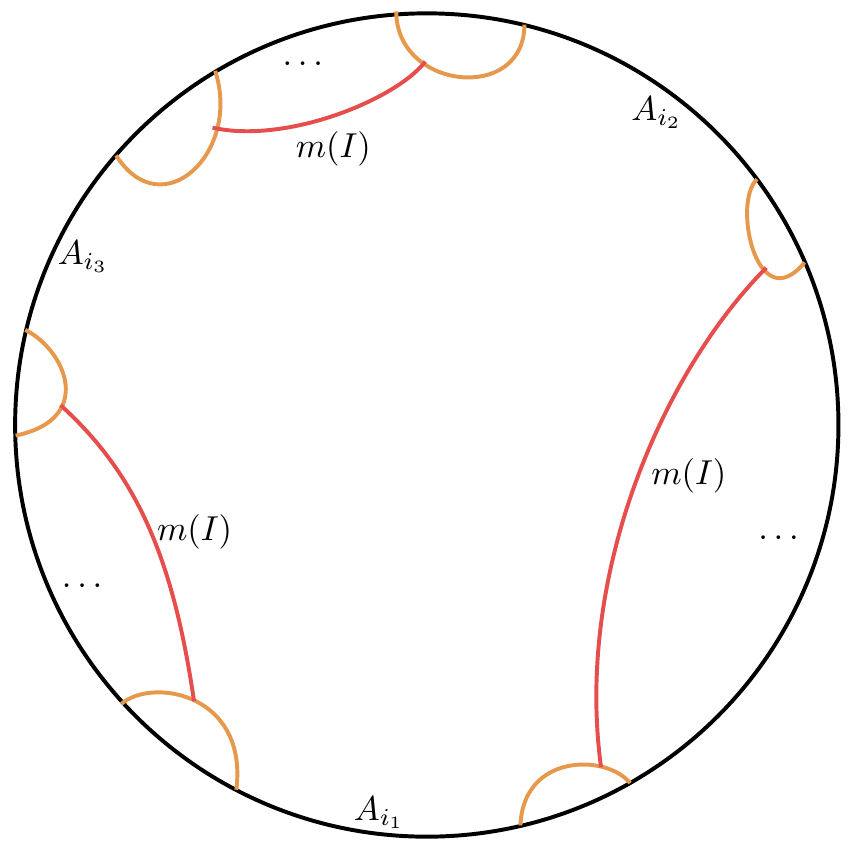}
        \label{RH_multiple}}
    \caption{(a) Relative homology for single boundary region $A$: the orange line
    represents relative surface $R$, the red lines $m(A)$ are homologous to $A$
    relative to $R$, and $r$ is the bulk region enclosed by $A, m(A)$ and $R$.
    (b) Relative homology for multiple boundary regions: orange lines still
    represent $R$, red lines $m(I)$ are homologous to $A_I$ relative to $R$, where
    $I=\{i_1,i_2,i_3\}$. We eliminate the bulk part that enclosed by $R$ and $\p M$,
    then we get $M^{\prime}$, i.e., the bulk region enclosed by $A$ and $R$.}
    \label{fig RHE}
\end{figure}

Let us now turn to the MFMC theorem associated with relative homology. This is a direct
generalization of the original MFMC theorem, that is \cite{Headrick:2017ucz}
\begin{equation}\label{gMFMC}
    \sup_{\substack{v^\mu:\\n_\mu v^\mu |_R=0}}
   \int_{A} \sqrt{h} n_\mu v^\mu=\inf_{\substack{m\sim A\\ \text{rel }R}} \ \text{Area}(m),
\end{equation}
where $h$ is the determinant of the induced metric on $A$, $n_\mu$ is the normal vector
on $A\cup R$,
and we impose a Neumann condition (no-flux condition) on $R$ for the flow.
Note that the MFMC theorem shows us the equivalency between ``max-flow" and ``min-cut".

Following RT formula, define a so-called relative homologous entropy (RHE) \cite{Headrick:2017ucz} as
\begin{equation}\label{RH entropy}
    \tilde{S}_A \coloneqq  \inf_{\substack{m\sim A\\ \text{rel }R}} \text{Area}(m)
    =\sup_{\substack{v^\mu:\\n_\mu v^\mu |_R=0}}
    \int_{A} \sqrt{h} n_\mu v^\mu.
\end{equation}

The above definition can be directly generalized to multiple boundary regions. Consider
$n$ non-overlapping boundary regions $A_i$ $(i\in[n]\coloneqq\{1,\dots,n\})$ on $\p M$.
We can choose a particular surface $R$ on $M$, whose boundary satisfying
$\p R=\bigcup _{i \in [n]}\p A_i$. So that $R$ and $A_i$ $(i\in[n])$ enclose an open
bulk region $M^{\prime}$, as shown in Fig.~\ref{RH_multiple}.

For boundary region
$A_{I}\equiv\bigcup_{i\in I}A_i$, where $I$ is a nonempty subset of $[n]$, we have
\begin{equation}\label{MRHE}
   \tilde{S}_I \coloneqq \inf_{\substack{m(I)\sim A_{I}\\ \text{rel }R}} \text{Area}(m(I))=\text{Area}[\tm(I)],
\end{equation}
where bulk surface $m(I)$ is homologous to $A_{I}\equiv\bigcup_{i\in I}A_i$ relative to $R$,
and $\tilde{m}(I)$ is defined as the minimal relative homologous (RH) surface of $A_I$.

For the following discussion,
we should mention that the status of $M^{\prime}$ in the RHE
inequality is the same as the status of $M$ in the HEE inequality.
The Neumann boundary condition $n_\mu v^\mu=0$
defined on $R$ implies that the flow (or bit thread) can not be allowed to get through $R$.
Thus a max flow emerging from $A_I$ must get through the minimum RH surface $\tilde{m}(I)$
and eventually arrive at $A_{[n]\backslash I}$, similarly for $A_{[n]\backslash I}$. So $\tilde{m}(I)$
and $\tilde{m}([n]\backslash I)$ are identical, it implies that RHE and HEE are extremely similar from
the inequality perspective.

And note that when $R$ is chosen to be empty set , RHE will become HEE \cite{FH, Headrick:2017ucz}.
Or we can choose $R$ to be min RT surface of some $A_I$, then RHE can reproduce holographic
entanglement wedge cross section proposal \cite{DCS, HHea}.
While in the following, we will study the general inequality properties of RHE
\eqref{MRHE} from the ``cut" side on a geometric perspective.


\section{Improved proof-by-contraction method}
\subsection{Brief review on proof-by-contraction method}
The RT formula reveals a deep connection between the bulk geometry and boundary
entanglement structures, which promotes us to study the inequalities of entanglement
entropy from a geometrical perspective.

A direct way is the following: if we can use some minimal RH surface of some union region
$A_{I_l}(l=1,2,\cdots,L)$ to get new RH surface (usually not minimal) of other union region
$A_{J_r}(r=1,2,\cdots,R)$, then we can have inequalities about $S_{I_l}$ and $S_{J_r}$.
See more specific descriptions in the following.

Minimum homologous surface of union region appearing on the left of inequality would be
cut into fragments by themselves since they may cross each other. Proof-by-contraction method is a way
to paste the surface fragments cut by themselves into homologous surface (usually not minimum) of union region
appearing on the right of inequality. Contraction map $f$ that we will introduce soon is the
key to how surface fragments are pasted.
The whole process can be simply written as cut-$f$-paste or just cut-paste.

Generally an inequality for RHE can be written as \cite{Akers:2021lms}
\begin{equation}
    \sum_{l=1}^{L}\alpha_l\tilde{S}_{I_l}\geq \sum_{r=1}^{R}\beta_r\tilde{S}_{J_r},
    \label{entropy inequality}
\end{equation}
where $\alpha_l,\beta_r$ are positive numbers
(usually they are rational number,
so we can multiply them by a large enough integer to get integer),
$L$ and $R$ are positive integers, representing the numbers of RHE terms appearing on each side.
$I_l,J_r \subseteq [n]$ and $\forall\ l\neq l^\prime, r\neq r^\prime$,
we have $I_l\neq I_{l^\prime},J_r\neq J_{r^\prime},I_l\neq J_{r}$. The only difference between
\eqref{rtentropy} and \eqref{entropy inequality} is that we have replaced the RT entropy
$S$ in \eqref{rtentropy} by the RHE $\tilde{S}$.

To construct the inequality, we need to cut the bulk $M^{\prime}$ by using
the minimum RH surface $\tm(I_l)$ of union regions $A_{I_l}$. So
these boundaries of bulk pieces are the fragments of  min RH surface.
The surface fragments can also be used to reconstruct the $\tilde{S}_{I_l}$ by pasting them back.

Then we try to paste these surface fragments to get RH surface of union regions
$J_r$. If we have inequality as follows,
\begin{equation}
    \sum_{l=1}^L\a _l \text{Area}[\tm(I_l)]\geq\sum_{r=1}^R\b _r \text{Area}[m^{\prime} (J_r)],
    \label{surface inequality}
\end{equation}
where $m^{\prime} (J_r)$ is the RH surface (usually not minimal) and we will prove the general inequality immediately. In what follows we will show how to find RH surfaces of $J_r$ for all $r=1,\cdots,R$ that satisfy inequality~\eqref{surface inequality}.


\subsection{Cut by left min RH surface}
To proceed, let us define an open bulk subregion $\sigma_I$ on $M^{\prime}$, called RH bulk region, which satisfies $\partial \sigma_I=\tm(I) \cup A_I$. We cut $M^{\prime}$ by $\tm(I_l)$ for all $l=1,\dots,L$. These $\tm(I_l)$ will cut $M^{\prime}$ into $2^L$ pieces (possibly including some
empty pieces). Then we encode these bulk pieces with length-$L$ bitstrings $x\in \{0,1\}^L$
\begin{equation}
    \sigma(x)\coloneqq\bigcap_{l=1}^{L}\sigma_{I_l}^{x_l},\qquad
    \sigma_I^b=\left\{\begin{matrix}
        \sigma_I & \text{if  } b=1,\\
        \sigma_I^{\cplm} &\text{if  }b=0,
    \end{matrix}\right.
    \label{bulk pieces def}
\end{equation}
where $x_l$ is the $l^{th}$ component of $x$ and $\sigma_I^{\cplm}$ is
the complement of $\sigma_I$ defined as $M^{\prime}\setminus\sigma_I$.
$\sigma(x)$ may be empty for some $x$.
It is easy to get $\forall x\neq \xp,\sigma(x)\cap \sigma(\xp)=\emptyset$ by
definition \eqref{bulk pieces def}.
The surface fragment can be written as
\begin{equation}
    \mxx\coloneqq\overline{\sigma(x)}\bigcap_{\text{codim-1}}
    \overline{\sigma(\xp)},
    \label{surface fragment def}
\end{equation}
where $\overline{\sigma(x)}$ represents the closure of the corresponding region $\sigma(x)$,
and  codim-1 means that we will calculate the area by using the integral measure of $\mxx$
in codimension-1 manifold. If dimension of $\mxx$ is less than codimension-1,
then the area of $\mxx$ is zero. We also need to require $x\neq \xp$,
otherwise $m(x,x)$ is just $\overline{\sigma(x)}$ whose dimension is codimension-0 and
the area of $m(x,x)$ will go to infinity. As shown in Fig.~\ref{SSA_cut},
the shaded bulk region corresponding to $\sigma_{I_1}$. And the $\sigma(x)$ are denoted by the
bitstrings $x$ in the figure, such as $"00"$ represents the open bulk region
$\sigma(00)$ and so on.

\begin{figure}[H]
    \centering
    \includegraphics[width=0.4\textwidth]{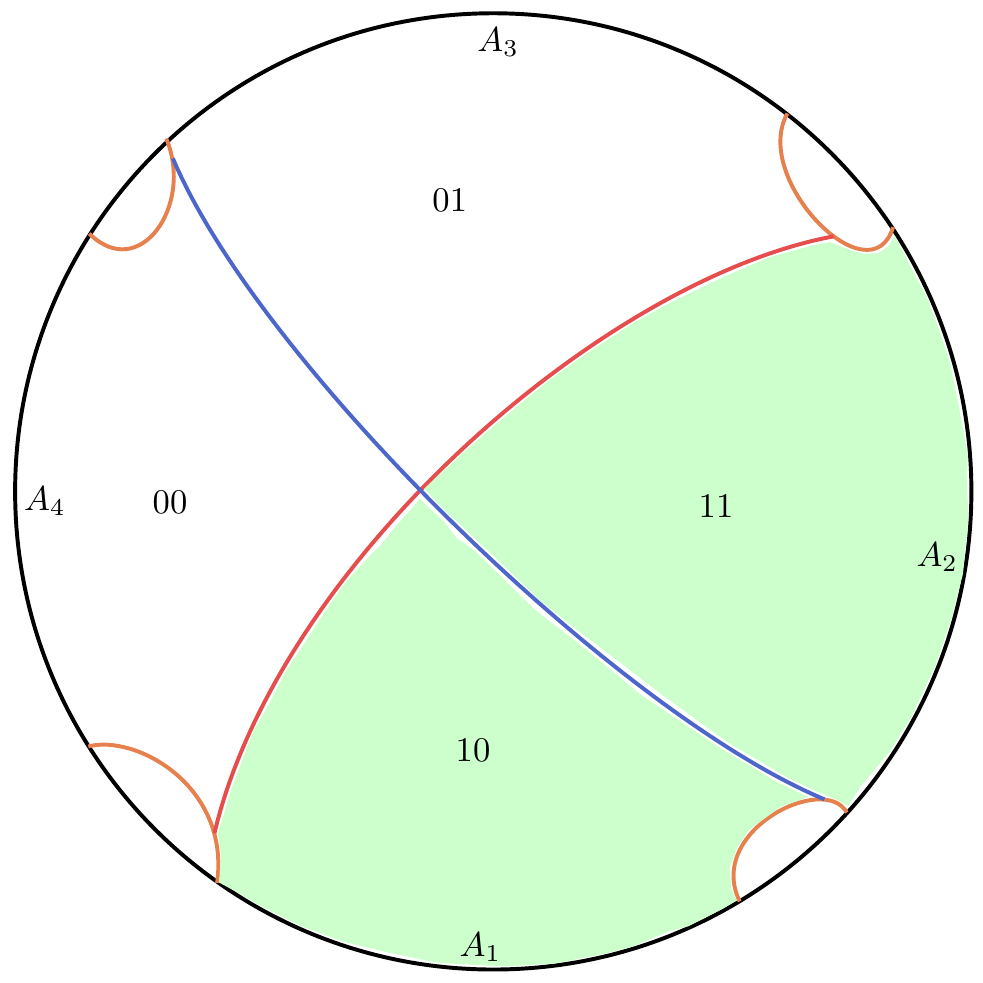}
    \caption{  Consider the SSA, i.e. $S_{A_1A_2}+S_{A_2A_3}\geq S_{A_2}+S_{A_1A_2A_3}$.
    The orange, red and blue lines represent the relative regions $R$, $\tm(I_1)$ and $\tm(I_2)$ respectively, where
    $I_1=\{1,2\}, I_2=\{2,3\}$. Where the shaded region denotes $\sigma_{I_1}$.  }
    \label{SSA_cut}
\end{figure}

By using these bulk pieces and surface fragments, we can reconstruct $\sigma_{I_l}$ and min RH surface
\begin{equation}
    \sigma_{I_l}=\bigcup_{x:x_l=1}\sigma(x).
    \label{sigil}
\end{equation}
This can be achieved directly from definition~\eqref{bulk pieces def}.
In addition, we can also get $\tm(I_l)$ as
\begin{align*}
    \tm(I_l)&\coloneqq\overline{\sigma_{I_l}}\bigcap_{\text{codim-1}}
        \overline{\sigma_{I_l}^\cplm}\\
        &=\left(\overline{\bigcup_{x:x_l=1}\sigma(x)}\right)\bigcap_{\text{codim-1}}
        \left(\overline{\bigcup_{\xp:\xp_l=0}\sigma(\xp)}\right) \\
        &=\left(\bigcup_{x:x_l=1}\overline{\sigma(x)}\right)\bigcap_{\text{codim-1}}
        \left(\bigcup_{\xp:\xp_l=0}\overline{\sigma(\xp)}\right) \\
        &=\bigcup_{x,\xp:x_l\neq \xp_l}\mxx.
\end{align*}
The second line stems from Eq.~\eqref{sigil}. In the fourth line we use the distribution law and the definition ~\eqref{surface fragment def}. While in the third line we have used that
\begin{align*}
    &\left\{\begin{array}{l}
    \overline{\bigcup\sigma(x)}
    =\partial\left(\bigcup\sigma(x)\right)\cup\left(\bigcup\sigma(x)\right)
    \subseteq\bigcup(\overline{\sigma(x)})\\
    \bigcup(\overline{\sigma(x)})=\bigcup(\partial\sigma(x)\cup\sigma(x))
    =\left(\bigcup\partial\sigma(x)\right)\cup\left(\bigcup\sigma(x)\right)
    \subseteq\overline{\bigcup\sigma(x)}
    \end{array}\right.\\
    \Rightarrow &\left\{\begin{array}{l}
        \overline{\bigcup\sigma(x)}\subseteq\bigcup(\overline{\sigma(x)})\\
        \bigcup(\overline{\sigma(x)})\subseteq\overline{\bigcup\sigma(x)}
        \end{array}\right.
    \Rightarrow \overline{\bigcup\sigma(x)}=\bigcup\overline{\sigma(x)}.
\end{align*}
In a word, we have
\begin{equation}
    \tm(I_l)=\bigcup_{x,\xp:x_l\neq \xp_l}\mxx.
\end{equation}
The RHE $\sil$ then can be obtained by definition
\begin{align}
    \sil&=\text{Area}[\tm(I_l)]
    =\sum_{\substack{x,\xp:\\x_l=1,\xp_l=0}}\text{Area}[\mxx]\notag\\
    &=\frac{1}{2}\sum_{x,\xp}|x_l-\xp_l|\text{Area}[\mxx].
    \label{sil}
\end{align}
Since $\forall (x^1,{x^1}^\prime)\neq (x^2,{x^2}^\prime)$ and
$\text{Area} [m(x^1,{x^1}^\prime)\cap m(x^2,{x^2}^\prime)]=0$,
we can directly add the area of $\mxx$ for different $\xx$ pair.
Due to the symmetry between $x$ and $\xp$, there is a prefactor $1/2$
in the last term.


\subsection{Paste into right RH surface}
We already know how to get surface fragments and reconstruct min RH surface using these fragments.
The purpose is constructing the RH surface of $J_r$ by these fragments.

The first step is to know the difference between RH surface and ordinary surface region.
It turns out that,  differing from ordinary surface region, $J_r$  and its RH surface form a bulk region
$\overline{\sigma_{J_r}}$ which only contains union region $A_{J_r}$. Mathematically, this corresponds to
\begin{align}
    &\overline{\sigma_{J_r}}\supseteq A_{J_r}=\bigcup_{i\in J_r}A_i \notag,\\
    &\overline{\sigma_{J_r}}\nsupseteq A_j(j \notin J_r)\text{\quad or\quad}\overline{\sigma_{J_r}^\cplm}
    \supseteq A_{[n]\setminus J_r}=\bigcup_{i\notin J_r}A_i.
    \label{constraint of sigjr}
\end{align}

Above restriction implies $\sigma_{J_r}$ is bounded. We take an example to specify the bound
as follow. Let us consider a monogamy-like inequality of RHE (the formal proof will be given in \ref{sect5.2}):
$\tS_{A_1A_2}+\tS_{A_1A_3}+\tS_{A_2A_3}\geq \tS_{A_1}+\tS_{A_2}+
\tS_{A_3}+\tS_{A_1A_2A_3}$, where $I_1=\{1,2\},I_2=\{1,3\},I_3=\{2,3\},J_1=\{1\}$. Then we
get the bound of $\sigma_{J_1}$ by Eq.~\eqref{constraint of sigjr}, as seen in
Fig.~\ref{MMI_bound}.

The next step is to construct $\sigma_{J_r}$. Only by using bulk pieces $\sigx$ can we construct
$\sigma_{J_r}$ to get $m^{\prime} (J_r)$. As suggested in \cite{Akers:2021lms}, we can connect
$\sigma_{J_r}$ and $\sigx$ by setting a contraction map $f:\{0,1\}^L\rightarrow\{0,1\}^R$. Similar to
Eq.~\eqref{sigil}, one can define $\sigma_{J_r}^{\prime}$
\begin{equation}\label{sigjr'}
    \sigma_{J_r}^{\prime}\coloneqq\bigcup_{x:f(x)_r=1}\sigma(x),
\end{equation}
where $f(x)_r$ is the $r^{th}$ component of $f(x)\in\{0,1\}^R$. However, up to now $\sigma_{J_r}^{\prime}$
represents nothing since no constraint has been imposed on $f$ yet. Recalling that we would like to
construct $\sigma_{J_r}$ from  $\sigx$. Thus, if we treat $\sigma_{J_r}^{\prime}$ as $\sigma_{J_r}$, this
can be achieved by definition \eqref{sigjr}. In this case, $\sigma_{J_r}^{\prime}$ must satisfy
constraint~\eqref{constraint of sigjr}, which in turn will impose some constraints on $f$. In what
follows, we would like to show how to put constraints on $f$.

\begin{figure}[H]
    \centering
    \subfigure[lower bound]{
    \includegraphics[width=0.4\textwidth]{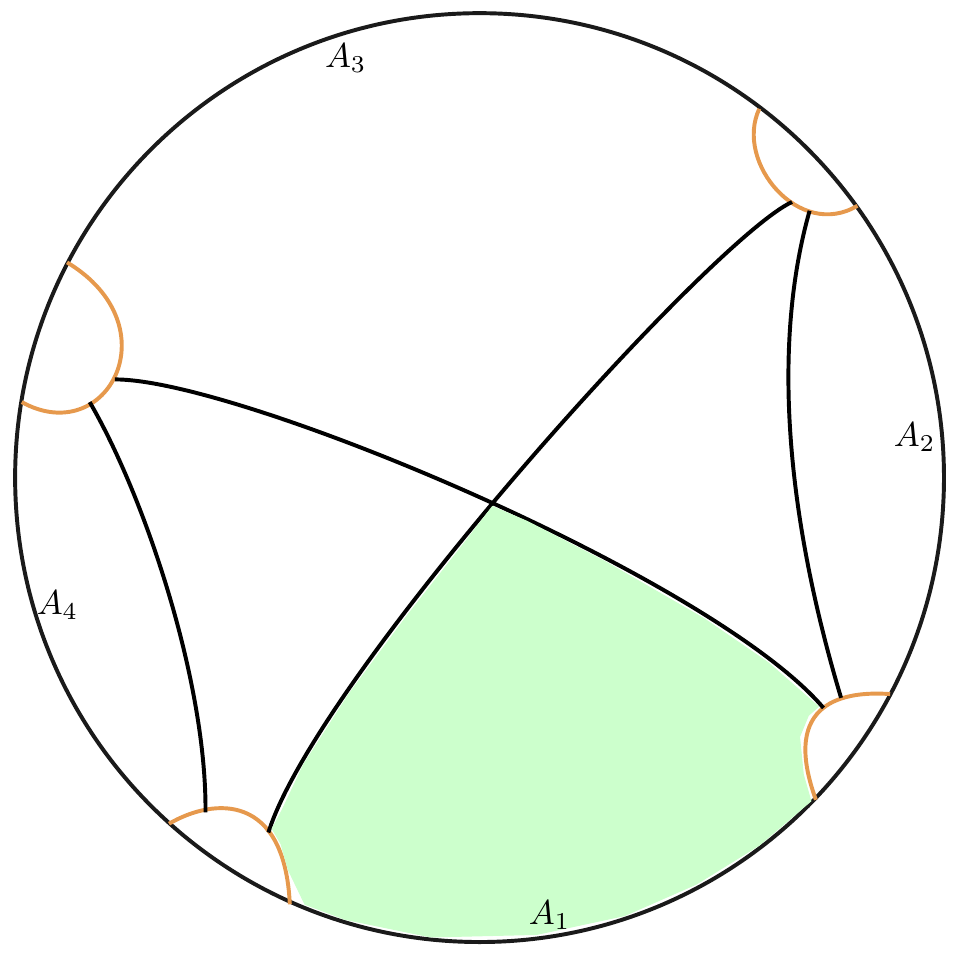}
    \label{lower bound}}
    \subfigure[upper bound]{
        \includegraphics[width=0.4\textwidth]{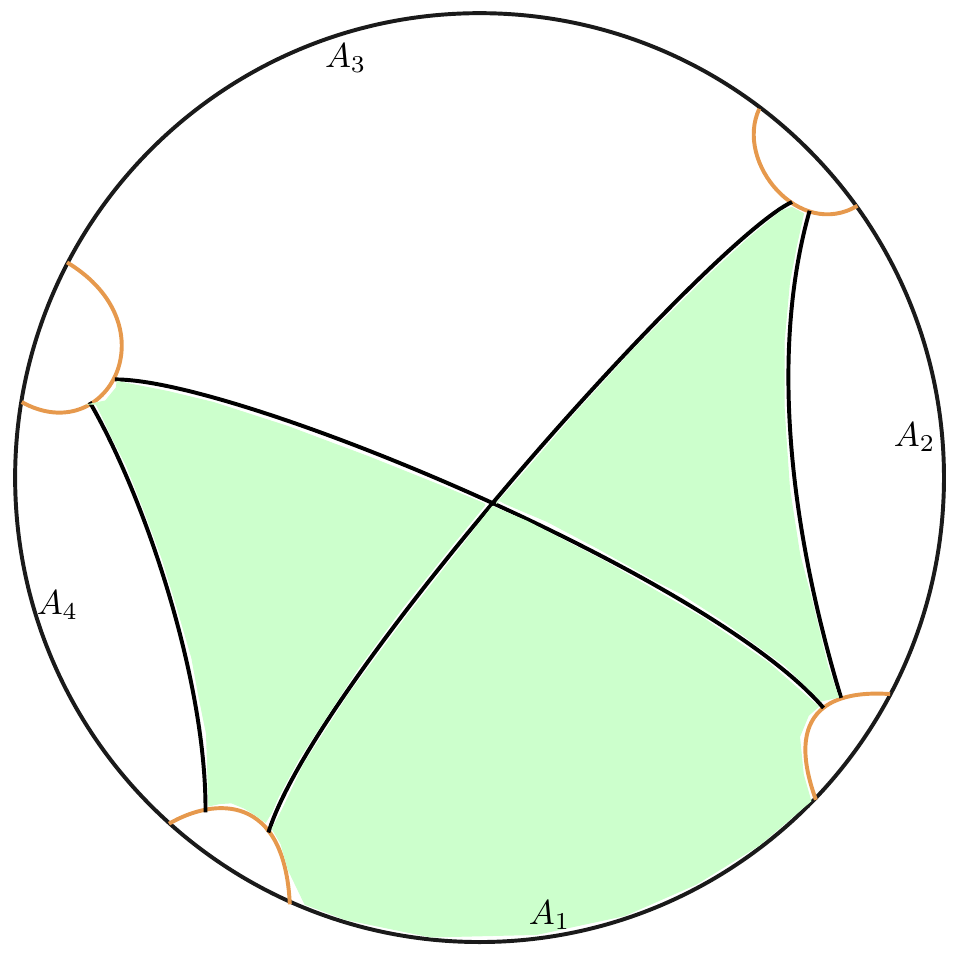}
        \label{upper bound}}
    \caption{The lower bound of $\sigma_{J_1}$ is shown in (a) and the upper bound of $\sigma_{J_1}$ is shown (b).
    The lower bound means intersection of all $\sigma_{J_r}$ defined by eq.~\eqref{sigjr} and the upper bound
    means union of all the $\sigma_{J_r}$.
    All $\overline{\sigma_{J_1}}$ bounded by eq.~\eqref{constraint of sigjr} contain RH surface of
    $A_1$. If $\overline{\sigma_{J_1}}$ is not bounded by eq.~\eqref{constraint of sigjr}, then
    $\p \sigma_{J_1}\setminus A_1$ will not correspond to RH surface of $A_1$.}
    \label{MMI_bound}
\end{figure}

So the third step is to add some basic constraints on $f$. First we should find all $\sigx$ that only contains $A_i$. Using $x^{(i)}$ to code the $\sigma(x^{(i)})$ that only contains $A_i$ and it can be written as
\begin{equation}
    x^{(i)}_l=\chi_{I_l}\equiv\left\{
        \begin{matrix}
            1 & i\in I_l,\\
            0 &i\notin I_l,
        \end{matrix}
    \right.
    \label{xi}
\end{equation}
where $\chi_{I_l}$ is a characteristic function defined by above equation.

Then we restrict $f$ so that $\sigma_{J_r}^{\prime}$ satisfies constraint~\eqref{constraint of sigjr}, that is,
\begin{equation}
    f(x^{(i)})_r=\chi_{J_r}\equiv\left\{
        \begin{matrix}
            1 & i\in J_r,\\
            0 & i\notin J_r,
        \end{matrix}
    \right.
    \label{fxi}
\end{equation}
where $\chi_{J_r}$ is also a characteristic function.

We will no longer distinguish $\sigma_{J_r}^{\prime}$ and $\sigma_{J_r}$ from now on,
since we have already required $\sigma_{J_r}^{\prime}$ to satisfy \eqref{constraint of sigjr} by restricting $f$. Thus we have
\begin{equation}
    \sigma_{J_r}=\bigcup_{x:f(x)_r=1}\sigma(x).
    \label{sigjr}
\end{equation}

Now we can follow what we did in getting $\tm(I_l)$ to obtain $m^{\prime} (J_r)$ immediately,
\begin{align*}
    m^\prime(J_r)&\coloneqq\overline{\sigma_{J_r}}\bigcap_{\text{codim-1}}
        \overline{\sigma_{J_r}^\cplm}\\
        &=\left(\overline{\bigcup_{x:f(x)_r=1}\sigma(x)}\right)\bigcap_{\text{codim-1}}
        \left(\overline{\bigcup_{\xp:f(\xp)_r=0}\sigma(\xp)}\right) \\
        &=\left(\bigcup_{x:f(x)_r=1}\overline{\sigma(x)}\right)\bigcap_{\text{codim-1}}
        \left(\bigcup_{\xp:f(\xp)_r=0}\overline{\sigma(\xp)}\right) \\
        &=\bigcup_{\substack{x,\xp:\\f(x)_r\neq f(\xp)_r}}\mxx.
\end{align*}

In the end, following eq~\eqref{sil}, we get the RHE for $J_r$,
\begin{align}
    \sjr&=\text{Area}[\tm(J_r)]\leq \text{Area}[m^\prime(J_r)]\notag\\
    &=\sum_{\substack{x,\xp:\\f(x)_r=1\\f(\xp)_r=0}}\text{Area}[\mxx]
    =\frac{1}{2}\sum_{x,\xp}|f(x)_r-f(\xp)_r|\text{Area}[\mxx].
    \label{sjr}
\end{align}


\subsection{Inequality of weight distance}\label{Inequality of weight distance}
Now we are stepping to show how it means if the inequality~\eqref{entropy inequality} holds.
From Eqs. \eqref{sil} and \eqref{sjr}, we can see that if
\begin{equation}\label{inequ1}
 \sum_{l=1}^L \alpha_l\frac{1}{2}\sum_{x,\xp}|x_l-\xp_l|
    \text{Area}[\mxx]\geq \sum_{r=1}^R \beta_r \frac{1}{2}\sum_{x,\xp}
    |f(x)_r-f(\xp)_r|\text{Area}[\mxx]
    \end{equation}
    holds, then the inequality~\eqref{entropy inequality} holds identically. While \eqref{inequ1} can be rewritten as
    \begin{equation}
    \sum_{x,\xp}\sum_{l=1}^L \alpha_l|x_l-\xp_l|
    \text{Area}[\mxx]\geq\sum_{x,\xp} \sum_{r=1}^R \beta_r
    |f(x)_r-f(\xp)_r|\text{Area}[\mxx].
\end{equation}
Therefore, if we define a weight distance $d_\omega(y,y^{\prime})$ by
\begin{equation}
    d_\omega(y,y^{\prime})\coloneqq\sum_i \omega_i\abs*{y_i-y^{\prime}_i},
\end{equation}
where $y_i,y^{\prime}_i$ are the $i^\text{th}$ component of $y,y^{\prime}$,
then we get
\begin{align}
    &\forall \xx\in \{\xx:\text{Area}[\mxx]\neq 0 \text{ and }x,x^{\prime} \in \{0,1\}^L\},\notag\\
    &d_\a(x,x^{\prime})\geq d_\b(f(x),f(x^{\prime}))
    \Rightarrow\sum_{l=1}^{L}\alpha_l\tilde{S}_{I_l}\geq \sum_{r=1}^{R}\beta_r\tilde{S}_{J_r}.
    \label{inequality equivalence}
\end{align}

So the problem has already been transformed from finding $m^{\prime}(J_r)$ to finding $f$ satisfying
basic constraints~\eqref{constraint of sigjr} and inequality of weight distance
$d_\a(x,x^{\prime})\geq d_\b(f(x),f(x^{\prime}))$ for all $\xx:\text{Area}[\mxx]\neq 0$, since
$m^{\prime}(J_r)$ is defined by $\sigma_{J_r}$ which is completely determined by $f$.
If we find $f$ that satisfies the conditions, we prove the inequality~\eqref{entropy inequality}.
We should emphasize, however, that if we cannot find $f$ satisfying the conditions
$d_\a(x,x^{\prime})\geq d_\b(f(x),f(x^{\prime}))$, it does not mean the inequality of RHE is
false. In other words, condition $d_\a(x,x^{\prime})\geq d_\b(f(x),f(x^{\prime}))$ is a sufficient
condition but is not a necessary one.

In the next section we will establish a general method to find $f$.


\section{Method of finding contraction map}
If an exhaustive method is used to find $f(x)$ for all $x$, the number of all situations of $f$ will be the order of $(2^R)^{2^L}$, which is an astronomical number. Only rule out most cases can we find $f$ in a relatively simple way. If we take a look at \eqref{inequality equivalence} carefully, it implies there exist three conditions for $f$ and $\xx$ needed to be obeyed, namely nonempty bulk pieces of $x$, $\text{Area}[\mxx]\neq 0$ and $d_\a(x,x^{\prime})\geq d_\b(f(x),f(x^{\prime}))$. Hence, it turns out one can find $f$ by the following steps:
\begin{itemize}
    \item First, we should exclude the bitstrings $x$ corresponding to empty bulk pieces. We will show
    that all bitstrings $x$ of empty bulk piece form a sub-set of all $x$ whose $g(x)$ is not equal to
    $\{1\}^{L^2}$, where $g(x)$ is defined in the following in Eq.~\eqref{gx}.
    \item Second, we should exclude the $\xx$ which corresponds to zero-area contribution, i.e., those
    corresponding to $\text{Area}[\mxx]= 0$. We will show that only for pairs $\xx$ with unit
    distance can $\mxx$ contribute nonzero area. The distance of $\xx$ denoted by $d \xx$ is
    defined as $d_\a \xx$, where $\a _l=1$, and $l=1,\dots,L$.
    \item Third, we then find $f(x)$ of $x$ corresponding to nonempty bulk piece. After requiring
    that $\forall x,x^{\prime}\in\mathbb{B}$, $f$ satisfies inequality of weight distance, one can finally
    find the map $f$. We will give the detailed steps to find these $f(x)$ in section \ref{find f}.
\end{itemize}


\subsection{Exclude empty bulk piece}

As a first step, let us define a mapping $g$ so as to distinguish whether $x$ correspond to empty bulk piece or
not. To proceed, we define $I^{0}\coloneqq[n]\setminus I,I^{1}\coloneqq I$. Then we define
mapping $g$ as follow
\begin{definition}
    Let $g$ be a map from $\{0,1\}^L$ to $\{0,1\}^{L^2}$ and
    $g(x)_{l,l^{\prime}}$ is the $[(l-1)L+l^{\prime}]^{\text{th}}$ component of $g(x)$, and
    \begin{equation}
        g(x)_{l,l^{\prime}}=\left\{
            \begin{matrix}
                1, & I_l^{x_l}\cap I_{l^{\prime}}^{x_{l^{\prime}}}\neq\emptyset,\\
                0,& I_l^{x_l}\cap I_{l^{\prime}}^{x_{l^{\prime}}}=\emptyset.
            \end{matrix}
        \right.
        \label{gx}
    \end{equation}
    \label{def gx}
\end{definition}
\begin{lemma}
    Min RH surfaces of two non-overlapping systems $A,B$ do not cross and the intersection of
    their RH bulk regions denoted as $\sigma_A\cap\sigma_B$ is empty.
    \label{intersection lemma}
\end{lemma}
\begin{proof}
    Suppose two non-overlapping systems $A,B$ and
    their min RH surfaces $\tm(A),\tm(B)$. Their RH regions are $\sigma_A,\sigma_B$ formed
    by $A,B$ and $\tm(A),\tm(B)$.

    Now let us assume that $\tm(A),\tm(B)$ will cross,
    then we get $\sigma_A\cap\sigma_B\neq\emptyset$ as shown in Fig.~\ref{RH cross}.  In this case,
    we find that $(\tm(A)\setminus\partial(\sigma_A\cap\sigma_B))\cup
    (\tm(B)\cap\partial(\sigma_A\cap\sigma_B))$ is still a RH surface of $A$.
    Without loss of generality, let us assume
    $\text{Area}[\partial(\sigma_A\cap\sigma_B)\cap\tm(B)]<
    \text{Area}[\partial(\sigma_A\cap\sigma_B)\cap\tm(A)]$, then
    $\text{Area}[(\tm(A)\setminus\partial(\sigma_A\cap\sigma_B))\cup
    (\tm(B)\cap\partial(\sigma_A\cap\sigma_B))]<\text{Area}[\tm(A)]$ which is
    contradiction with the fact that $\tm(A)$ is the min RH surface of $A$. The discussion is same for $B$.
\end{proof}

\begin{figure}[H]
    \centering
    \includegraphics[width=0.4\textwidth]{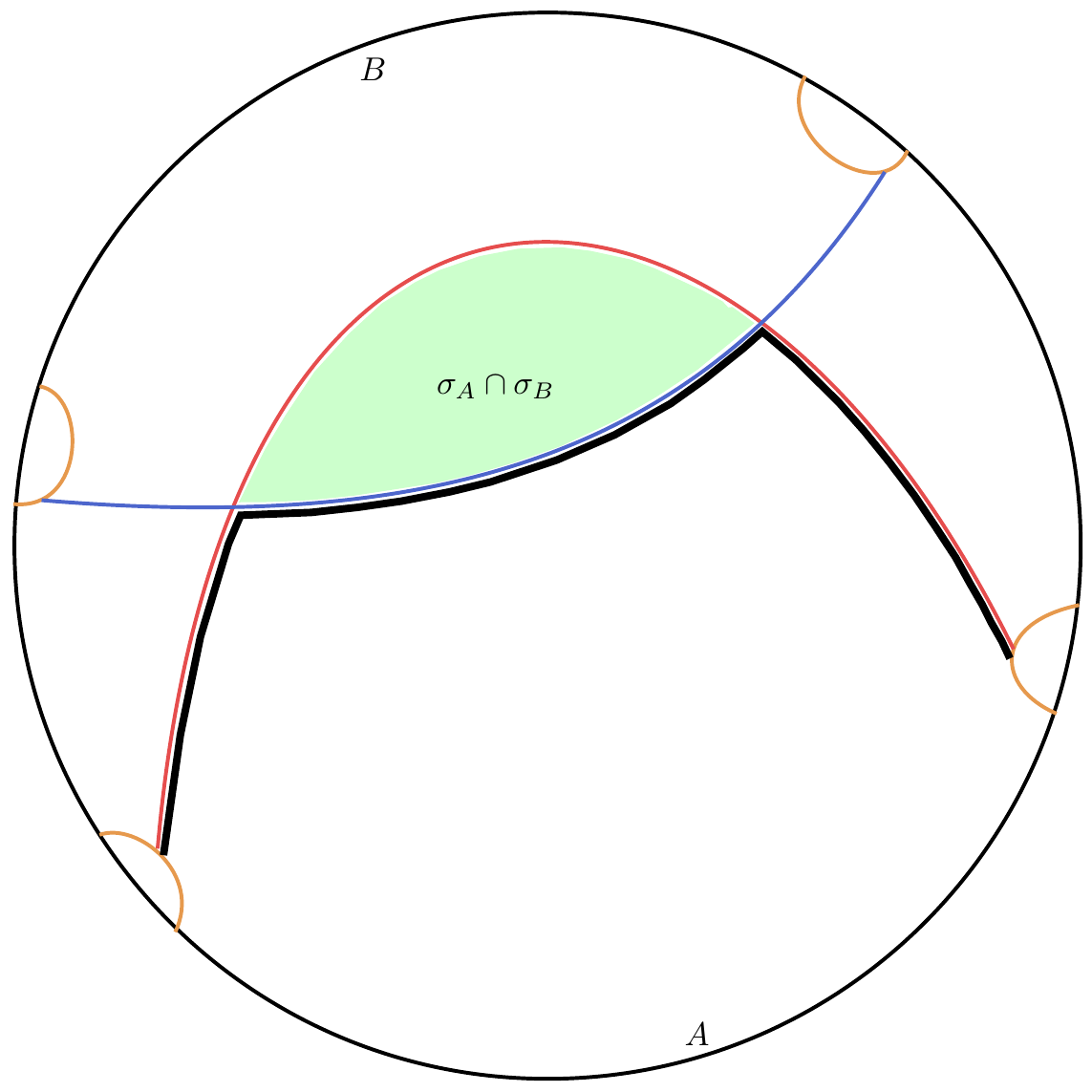}
    \caption{The red line is $\tm(A)$, the blue line is $\tm(B)$ and the green shadow is
    $\sigma_A\cap\sigma_B$. And the black outline is $(\tm(A)\setminus\partial
    (\sigma_A\cap\sigma_B))\cup(\tm(B)\cap\partial(\sigma_A\cap\sigma_B)$.}
    \label{RH cross}
\end{figure}

\begin{theorem}
    $x$ corresponds to an empty bulk piece if $g(x)$ is not equal to $\{1\}^{L^2}$, where
    $\{1\}^{L^2}$ means $\forall\ l,l^{\prime}, g(x)_{l,l^{\prime}}=1$.
    \label{gx theorem}
\end{theorem}

\begin{proof}
    According to Lemma.~\ref{intersection lemma}, we can get that $\tm(I)$ and $\tm(I^{\prime})$
    do not intersect when $I\cap I^{\prime}=\emptyset$. In addition,
    $\tm(I)$ and $\tm(I^{\prime})$ still do not intersect if we have
    $I\cap ([n]\setminus I^{\prime})=\emptyset$ instead of $I\cap I^{\prime}=\emptyset$,
    since $\tm(I)=\tm([n]\setminus I)$ holds by choosing
    $R=\{R\in M:\partial R=\bigcup_{i\in[n]}\partial A_i\}$.

    If $g(x)\neq \{1\}^{L^2}$, then $\exists\ l,l^{\prime}, g(x)_{l,l^{\prime}}=0$.
    According to Definition.~\ref{def gx} and above statement,
    we will get $\sigx\subseteq \sigma_{I_l}^{x_l}\cap \sigma_{I_{l^{\prime}}}^{x_{l^{\prime}}}=\emptyset$,
    which exactly means that $x$ corresponds to an empty bulk piece.
\end{proof}

We will denote the supset of all $x$ encoding nonempty bulk pieces by $\mathbb{B}$.
In other words, for every $x$ encoding nonempty bulk pieces, $x\in \mathbb{B}$.
That is, we define $\mathbb{B}$ as
\begin{equation}
    \mathbb{B}\coloneqq\{x:g(x)=\{1\}^{L^2}\text{ and } x\in\{0,1\}^{L}\}.
\end{equation}

In summary, we only need to find $f(x)$ for nontrivial
$x\in \mathbb{B}$, since $f(x)$ for trivial
$x\in \{x:g(x)\neq\{1\}^{L^2}\}$ will not affect inequality~\eqref{entropy inequality}.


\subsection{Exclude surface fragment of zero-area contribution} \label{4.2}

Now we would like to find all pair $\xx$ whose corresponding $\mxx$ with nonzero area contribution.
Such $\xx$ have a close relationship with $d\xx$ as we will see.

\begin{definition}
    Let $\sigma_{I_l}(l=1,\dots,L)$ be an open set and we can generate a topology $\mathcal{T}$
    that is a set of open subsets of $M^\prime$ satisfying:
    \begin{itemize}
        \item $\emptyset,M^{\prime},\sigma_{I_1},\dots,\sigma_{I_L}\in\mathcal{T}$,
        \item $\forall \a\in\Lambda$ where $\Lambda$ is some either finite or infinite set,
        $O_\a\subset M^{\prime}$,
        if $\{O_\a\}_{\a\in\Lambda}\subseteq\mathcal{T}$,
       then $\bigcup_{\a\in\Lambda}O_\a \in \mathcal{T}$,
        \item $m\in \mathbb{N}$, $\forall i\in[m], O_i\subset M^{\prime}$,
        if $\{O_i\}_{i\in[m]}\subseteq\mathcal{T}$,
        then $\bigcap_{i=1}^m O_i \in\mathcal{T}$.
    \end{itemize}
    \label{topology}
\end{definition}

\begin{definition}
Let $\xx$ be a pair where $x,x^{\prime} \in \mathbb{B}$.
$\xx$ is {\bf directly connected} if and only if only one component of $x$ and $x^{\prime}$ is different.
$\xx$ is {\bf connected} if there is a path, for all $x^i\in \mathbb{B},(i=1,\cdots,n)$ on the path,
$(x^i,x^{i+1})$ is always directly connected where $x^1=x$ and $x^n=x^{\prime}$.
\label{connection}
\end{definition}

For example, we can see Fig.~\ref{SSA_cut} where $(00,01),(00,10),(01,11),(10,11)$ are all
directly connected. And according to Definition.~\ref{connection},
$(00,11),(01,10)$ are connected since there are paths $00\rightarrow 01\rightarrow 11$ and
$01\rightarrow 00\rightarrow 10$ connecting $00,11$ and $01,10$ where $(00,01),(01,11)$ and
$(01,00),(00,10)$ are directly connected. We should mention that two geometrically joint regions
$\sigx$ and $\sigma(\xp)$ do not always mean that the corresponding pairs $\xx$ are directly
connected and vice versa. One explicit example is depicted in Fig. \ref{directly_connected}.

\begin{figure}[H]
  \centering
  \subfigure[]{
    \includegraphics[width=0.4\textwidth]{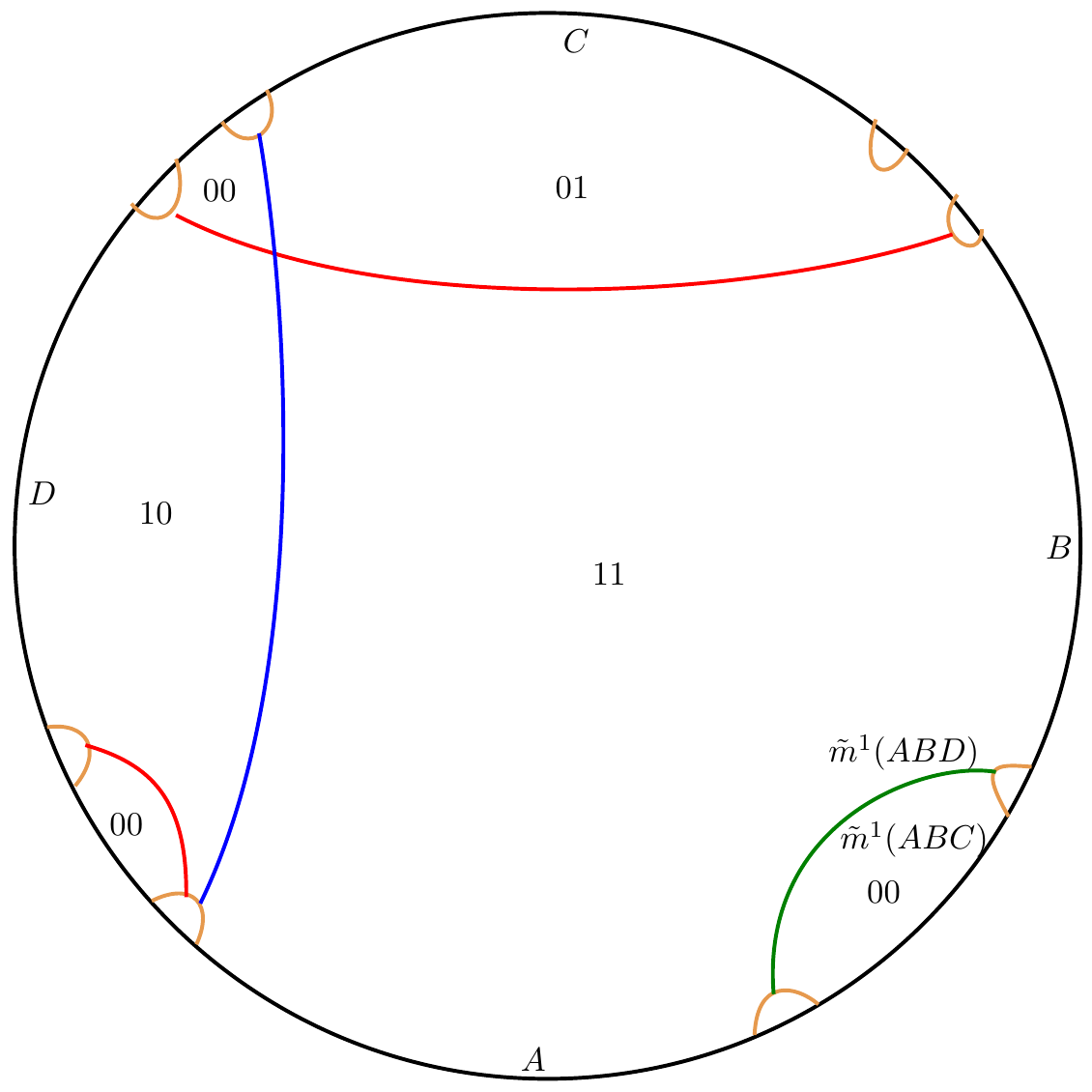}}
    \subfigure[]{
        \includegraphics[width=0.44\textwidth]{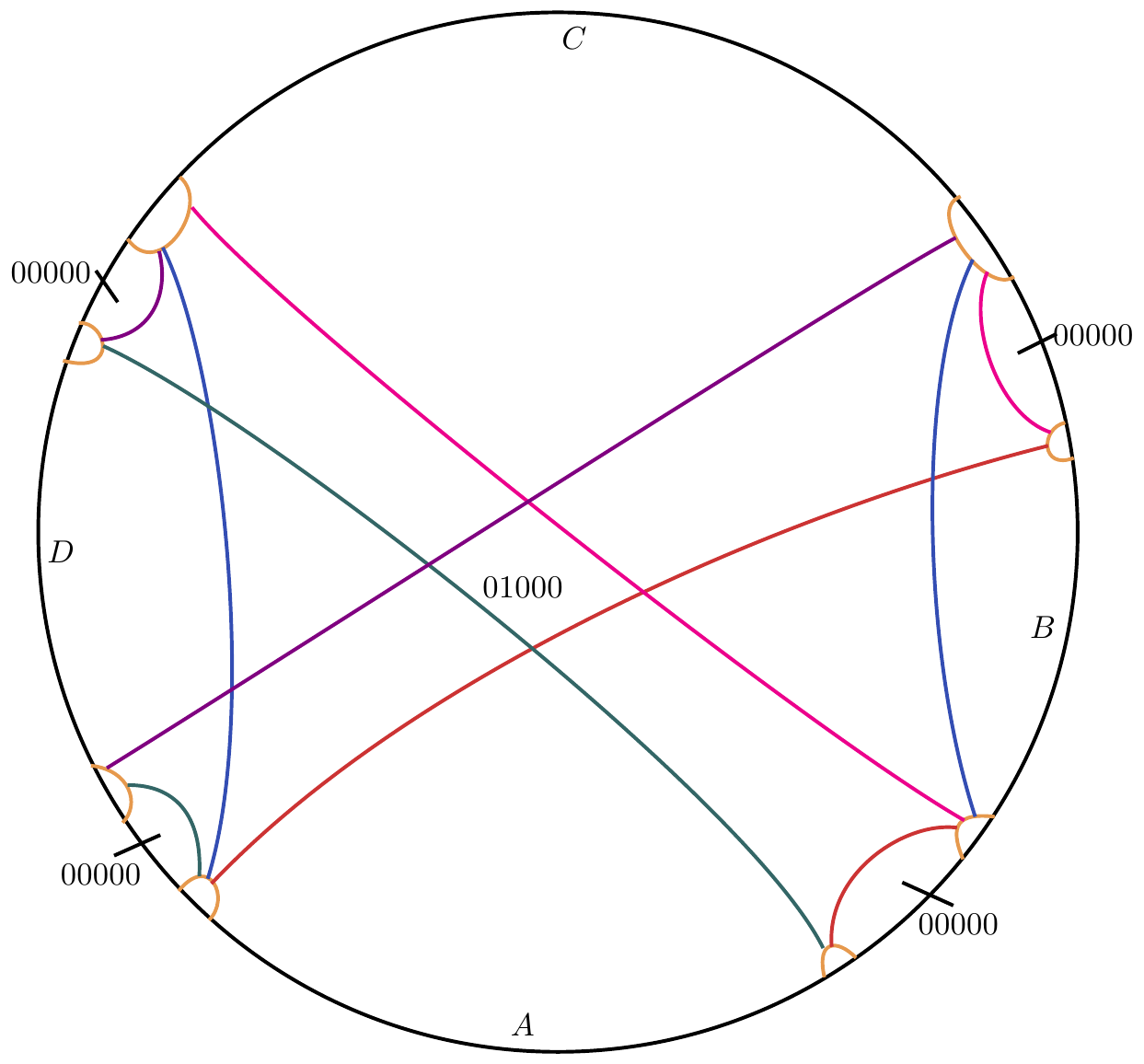}}
  \caption{(a) Example that  $\sigx$ and $\sigma(\xp)$ are geometrical
  neighbors while the corresponding pairs $\xx$ are NOT directly connected.
  Let $I_1=ABC,I_2=ABD$, then $\sigma(00)$ and $\sigma(11)$ are geometrical
  neighbors, while the corresponding pairs ($x=00,x'=11$) are not directly connected.
  Note that in this case the surface fragments $\tm^1(ABC)$ and $\tm^1(ABD)$ are
  overlapped and are marked by green line.
  (b) Example that $\xx$ are directly connected while the corresponding relative
  regions $\sigx$ and $\sigma(\xp)$ are NOT geometrical neighbors. Let
  $I_1=AB,I_2=AC,I_3=AD,I_4=BC,I_5=CD$, then $(00000,01000)$ is directly
  connected, while $\sigma(01000)$ and $\sigma(00000)$ are not geometrical
  neighbors.}
  \label{directly_connected}
\end{figure}

\begin{definition}
   Let $m_l(x)$ encoded by $x\in\{0,1\}^L$ be a fragment of $\tm(I_l)$ cut by $\tm(I_{i}),i\in[L]\setminus\{l\}$
    and $x^{\prime}_i=(1-\delta_{il}) x_i$\footnote{We do not use sum rule here. Namely, we do not sum over
    $i$ in this formula.}, $x_l=1$. Then $m_l(x)$ is defined as
    \begin{equation}
        m_l(x)\coloneqq\bigcap_{\substack{i=1\\i\neq l}}^L \sigma_{I_i}^{x_i}\cap \tm(I_l)
        =[\sigma(x)\cup\sigma(x^{\prime})]\cap\tm(I_l).
        \label{fragment of mil}
    \end{equation}
    \label{mlx}
\end{definition}
We can easily find that the definition of $m_l(x)$ is irrelevant to $x_l$, so the pair $\xx$
in which only the $l^{th}$ component is different can also correspond to the fragment of $\tm(I_l)$.
Please see Fig.~\ref{cut mil} for a more intuitive understanding.

According to Definition.~\ref{topology}, Eq.~\eqref{fragment of mil} can be written as
\begin{align}
    m_l(x)=&\left(\bigcap_{x_i=1,i\neq l}\sigma_{I_i}\right)
    \cap\left(\bigcup_{x_j=0,j\neq l}\sigma_{I_j}\right)^\cplm\cap \tm(I_l)\notag\\
    =&O_\a(x,l)\cap O_\b(x,l)^\cplm\cap \tm(I_l),
\end{align}
where
\begin{align*}
    O_\a(x,l)\coloneqq\bigcap_{x_i=1,i\neq l}\sigma_{I_i}\in \mathcal{T},
    O_\b(x,l)\coloneqq\bigcup_{x_i=0,i\neq l}\sigma_{I_i}\in \mathcal{T}.
\end{align*}

So far we have known that all surface fragments $\mxx$ correspond to $\xx$. The next question is what the
relationship between $\mxx$ and $m_l(x)$ should be. The following lemma tells us their areas
are equal only if $\xx$ is directly connected.

\begin{lemma}
    If $\xx$ is directly connected, then $\text{Area}[\mxx]=\text{Area}[m_l(x)]$.
    \label{area lemma}
\end{lemma}

\begin{proof}
    \begin{align*}
        &[\sigma(x)\cup\sigma(x^{\prime})]\cap\sigma_{I_i}=\sigma(x)\\
        &[\sigma(x)\cup\sigma(x^{\prime})]\cap\sigma_{I_i}^\cplm=\sigma(x^{\prime})\\
        \Rightarrow\overline{\sigma(x)}\cap\overline{\sigma(x^{\prime})}
        &=\overline{[\sigma(x)\cup\sigma(x^{\prime})]\cap\sigma_{I_i}}\cap
        \overline{[\sigma(x)\cup\sigma(x^{\prime})]\cap\sigma_{I_i}^\cplm}\\
        &\subseteq\overline{\sigma(x)\cup\sigma(x^{\prime})}\cap\overline{\sigma_{I_i}^\cplm}
        =\overline{\sigma(x)\cup\sigma(x^{\prime})}\cap\tm(I_i)\\
        \Rightarrow\overline{\sigma(x)}\cap\overline{\sigma(x^{\prime})}&\setminus
        [(\sigma(x)\cup\sigma(x^{\prime}))\cap\tm(I_i)]\\
        &\subseteq[\overline{\sigma(x)\cup\sigma(x^{\prime})}\cap\tm(I_i)]
        \setminus[(\sigma(x)\cup\sigma(x^{\prime}))\cap\tm(I_i)]\\
        &\subseteq[\overline{\sigma(x)\cup\sigma(x^{\prime})}
        \setminus(\sigma(x)\cup\sigma(x^{\prime}))]\cap\tm(I_i)\\
        &\subseteq\p(\sigma(x)\cup\sigma(x^{\prime}))\cap\tm(I_i)\\
        &\subseteq\left(\bigcup_{j\neq i}\tm(I_j)\right)\cap\tm(I_i)
        =\bigcup_{j\neq i}\left[\tm(I_j)\cap\tm(I_i)\right].
    \end{align*}

    For $\forall i\neq j$, the dimension of $\tm(I_i)\cap\tm(I_j)$ is codimension-2
    (If there is some overlapping region of $\tm(I_l)$, then $\tm(I_i)\cap\tm(I_j)$ may be codimension-1.
    We will discuss this situation in section \ref{discuss}), where its area contribution is zero.
    So $\mxx$ and $m_l(x)$ only have the difference in zero-area contribution region.
\end{proof}
This lemma can be understood intuitively by seeing Fig.~\ref{cut mil}.

\begin{figure}[H]
    \centering
    \includegraphics[width=0.4\textwidth]{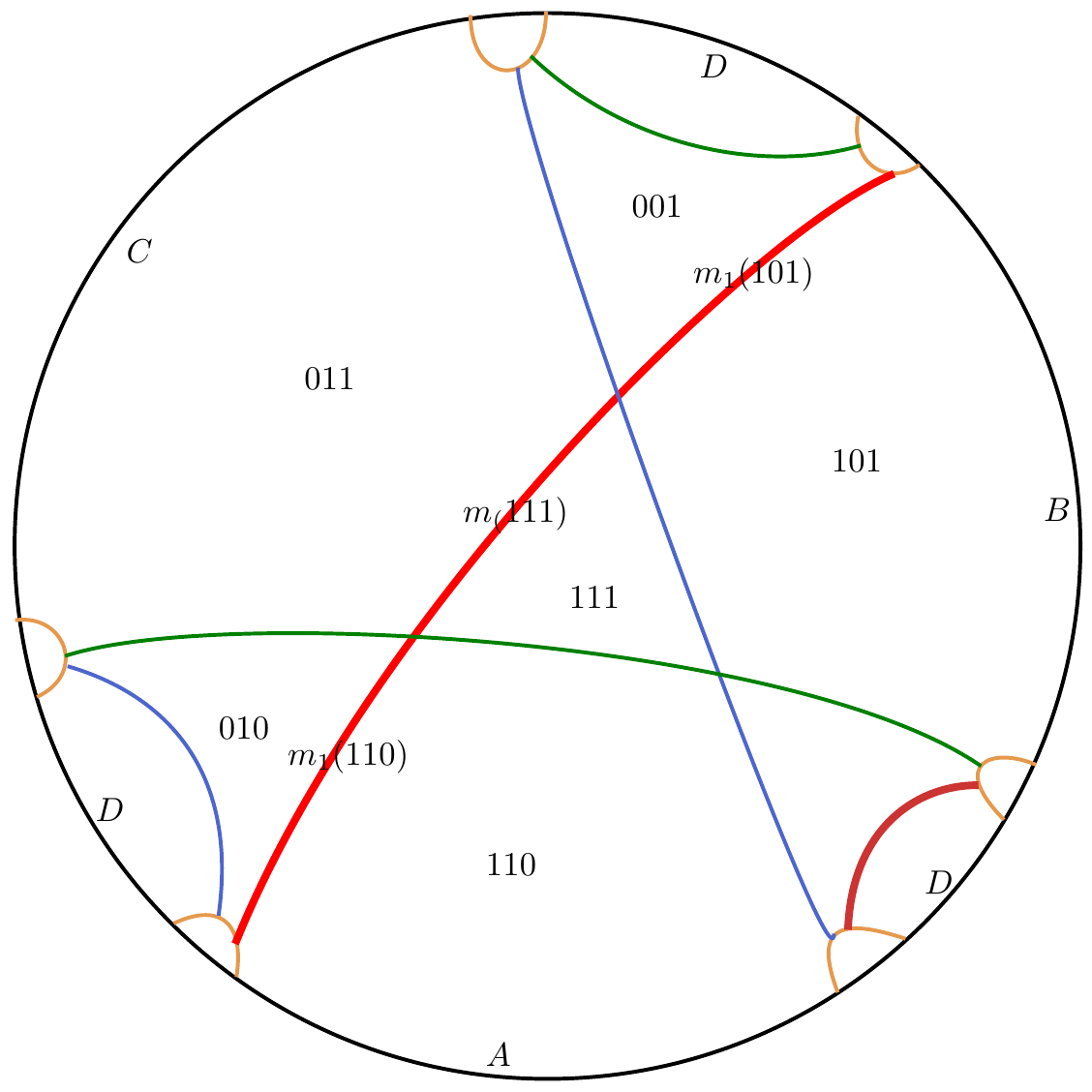}
    \caption{We choose $I_1,I_2,I_3$ as $AB,AC,BC$ as before.
    The bold red line is $\tm(AB)$, blue and green lines are $\tm(AC)$ and $\tm(BC)$ respectively.
    $\tm(AB)$ is cut into three parts by $\tm(AC)$ and $\tm(BC)$. We can see that $m_1(110)$,
    $m_1(111)$ and $m_1(101)$ are almost the same as $m(010,110)$, $m(011,111)$ and $m(001,101)$
    respectively.}
    \label{cut mil}
\end{figure}

\begin{theorem}
    Only if $\xx$ is directly connected can $\mxx$ contribute nonzero area.
    \label{nonzero area theorem}
\end{theorem}
\begin{proof}
    According to Definition.~\ref{mlx}, only $m_l(x)$ can contribute nonzero area and
    we have $\text{Area}[\mxx]=\text{Area}[m_l(x)]$ if $\xx$ is directly connected according to
    Lemma.~\ref{area lemma}. So only when $\xx$ is directly connected can $\mxx$
    contribute nonzero area.
\end{proof}
Then we replace $m_l(x)$ with surface fragment $\mxx$. Since all fragments of
$\tm(I_l),l=1,\dots,L$ correspond to $\xx$, only when $d\xx=1$ might $\mxx$ have
nonzero contribution. We denote the set of all directly connected $\xx$ by $\mathbb{A}$,
which means that these surface fragments have nonzero area contribution. In other words, we have
\begin{equation}
    \mathbb{A}\coloneqq\{\xx:d\xx=1\text{ and }x,x^{\prime}\in \mathbb{B}\}.
\end{equation}

\begin{theorem}
    Direct connection is topological invariant under $\mathcal{T}$.
    \label{DC theorem}
\end{theorem}
\begin{proof}

We find that all $\xx$ satisfying $\forall i\in[L]\setminus\{l\}, x^{\prime}_i=x_i$ and $x^{\prime}_l\neq x_l$
will correspond to $(O_\a(x,l),O_\b(x^{\prime},l))$ according to the definition of $O_\a(x,l)$ and
$O_\b(x,l)$ in Definition \ref{mlx}. And since the definition of $O_\a(x,l)$ and $O_\b(x,l)$ is
irrelevant to $x_l$, we should treat $(O_\a(x,l),O_\b(x^{\prime},l))$ and
$(O_\a(x^{\prime},l),O_\b(x,l))$ as the same.

Then we can think that there is a map $\phi$ that $\phi(O_\a(x,l),O_\b(x^{\prime},l))=(x,x^{\prime})$.
We can easily verify that the map $\phi$ is well-defined and one-to-one if we treat $\xx$ and
$(x^{\prime},x)$ as the same. So directly connection is topological invariant under $\mathcal{T}$.
\end{proof}

Changing the systems' position on $\p M$
and occupied region does not change the structure of the topology $\mathcal{T}$.
So the direct connection of $\xx$ is irrelevant to systems' position and occupied region.

\begin{theorem}
    For all $x,x^{\prime}\in \mathbb{B}$, $\xx$ is connected.
    \label{C theorem}
\end{theorem}

\begin{proof}
    For all $x,x^{\prime}\in \mathbb{B}$, since $M^{\prime}$ is simply connected, there must be a path in
    $M^{\prime}$ that connects $\sigma(x)$ and $\sigma(x^{\prime})$ passing through some $\tm(I_l)$.
    Each time when the path passes through a
    certain $\tm(I_l)$, the bitstrings $x$ will change and only change the component $x_l$
    once. As shown in Fig.~\ref{connected path}, each bitstring $x$ is associated with a superscript
    $i$ whose value is arranged in order. Each time when the path gets through $\tm(I_l)$, $x^i$
    changes to $x^{i+1}$, i.e., the value of $i$ increases by one. In this way we get $x^{i} (i=1,\cdots,n)$
    in proper order along the path and we let $x^1=x, x^n=x^{\prime}$.  Therefore, $(x^i,x^{i+1})
    (i=1,\cdots,n-1)$ is directly connected by definition. And thus,  according to Definition.~\ref{connection},
    $\xx$ is connected.
    \end{proof}

\begin{figure}[H]
    \centering
    \includegraphics[width=0.4\textwidth]{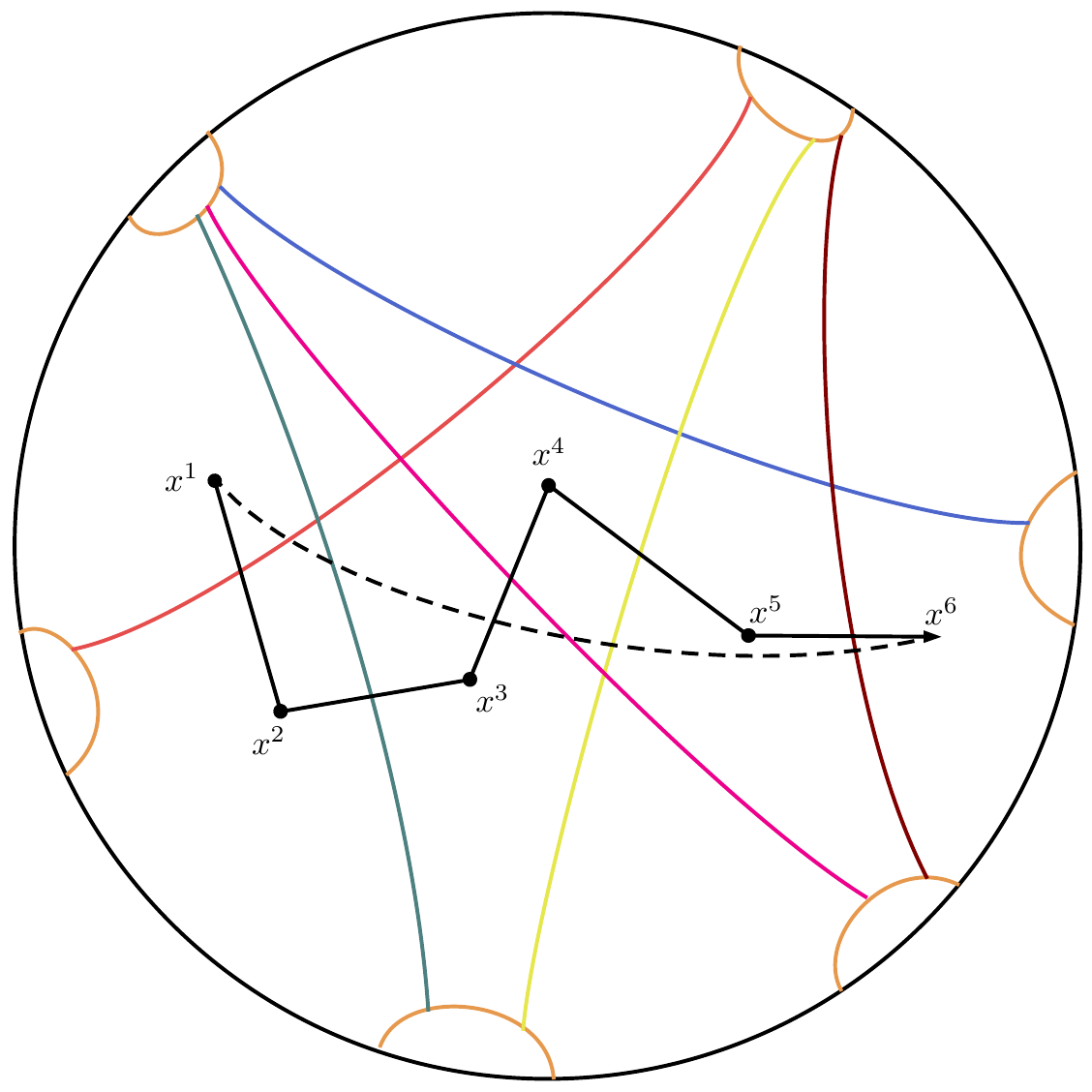}
    \caption{Polygonal line with arrow is the path that connect $\sigma(x^1)$ and $\sigma(x^6)$,
    and dotted line is the shortest (geometrical) path between the points of $\sigma(x^1)$ and $\sigma(x^6)$.
    We can easily find that $(x^i,x^{i+1})$ is directly connected. The polygonal (black bold) line is the
    corresponding shortest weight path, which will pass
    through the $\tm(I_l)$(color lines) when it moves from $\sigma(x^i)$ to $\sigma(x^{i+1})$.}
    \label{connected path}

\end{figure}


\subsection{Find all $f(x)$ for nonempty bulk pieces}\label{find f}
In previous subsections, we have shown how to exclude pairs with empty bulk pieces and surface fragments of zero-area
contribution. Although by this way one can exclude most surface fragments, the situations for finding $f$ are still huge.
Let us explain in the following. Before imposing the
zero-area-contribution constraint, the number of allowed choices for $f(x)$ is
$N(\mathbb{B})$ (same as the number of $x\in\mathbb{B}$),
where $N(\mathbb{B})$ represents the total number of elements in
$\mathbb{B}$. And the number of $f$'s inequality constraints is $N(\mathbb{A})$.
Even though the number of inequality constraints is small enough,
if we use an exhaustive method to find $f$, there is still $2^R$ possible $f(x)$ for each
uncertain $x$. The number of uncertain $x$ is $N(\mathbb{B})-n$, where $n$ is the number of $f(x^{(i)})$ which has been
fixed by the basic constraints \eqref{fxi}. So the number of possible $f$ is
$(2^R)^{N(\mathbb{B})-n}$. In addition, we should verify $N(\mathbb{A})$ times at most
for each possible $f$. Totally we need to verify $N(\mathbb{A})(2^R)^{N(\mathbb{B})-n}$ times
at most to find $f$. This is undoubtedly still huge. However, if we require that $f$ satisfying
$d_\b\fxx\leq d_\a\xx$ for all $x,\xp \in\mathbb{B}$\footnote{Note that this constraint
imposed on $f$ is weaker than the original one as imposed in \cite{Bao:2015bfa}, where,
instead of $x,\xp \in\mathbb{B}$, they require $f$ satisfying $d_\b\fxx\leq d_\a\xx$ for all
$x,\xp \in\{0,1\}^L$.}, it turns out that one can find the possible $f(x)$ for some certain $x$
programmatically.

Our algorithm includes two parts. Firstly, 
let us denote the set of all possible $f(x)$ as $\mathcal{F}^1(x)$, then $\mathcal{F}^1(\ix)=\{\fix\}$.
Define $\mathbb{B}_0\coloneqq\bigcup_{i\in[n]}\{\ix\}$ and let $\Bp=\mathbb{B}_0$.
See the detailed processes as follows:
\begin{itemize}
    \item First step. Choose $x^1\in\mathbb{B}\setminus\Bp$, then find all possible $f(x^1)\in\{0,1\}^R$
    satisfying $d_\a (x^1,\xp)\geq d_\b(f(x^1),f(\xp))$ for all $f(\xp)\in\mathcal{F}^1(\xp)$ of all
    $\xp\in\Bp$. Denote the set of all possible $f(x^1)$ as $\mathcal{F}^1(x^1)$, then add $x^1$
    into $\Bp$, i.e. resetting $\Bp=\{x^1\}\cup\Bp$.
    \item Second step. Choose $x^2\in\mathbb{B}\setminus\Bp$, then find all possible $f(x^2)\in\{0,1\}^R$
    satisfying $d_\a (x^2,\xp)\geq d_\b(f(x^2),f(\xp))$ for all $f(\xp)\in\mathcal{F}^1(\xp)$ of all
    $\xp\in\Bp$. Denote the set of all possible $f(x^2)$ as $\mathcal{F}^1(x^2)$, then add $x^2$
    into $\Bp$, i.e. resetting $\Bp=\{x^2\}\cup\Bp$.
    \item ......
    \item $m^{\text{th}}$ step. Choose $x^m\in\mathbb{B}\setminus\Bp$, then find all possible $f(x^m)\in\{0,1\}^R$
    satisfying $d_\a (x^m,\xp)\geq d_\b(f(x^m),f(\xp))$ for all $f(\xp)\in\mathcal{F}^1(\xp)$ of all
    $\xp\in\Bp$. Denote the set of all possible $f(x^m)$ as $\mathcal{F}^1(x^m)$, then add $x^m$
    into $\Bp$, i.e. resetting $\Bp=\{x^m\}\cup\Bp$.
\end{itemize}
See Figure.~\ref{MMI_graph} to help understanding the above steps intuitively.

If the number of possible $f(x)$ is large, we can narrow the range of $\mathcal{F}^1(x)$ by
repeating the above steps with slight changes. Because there already have been $\mathcal{F}^1(x)$
for every $x$, we should consider $\xp\in\mathbb{B}$. And $x\in\mathbb{B}\setminus\mathbb{B}_0$
is also already numbered and does not need to be specified separately. Take $\forall x\in\mathbb{B}$ and $\mathcal{F}^2(x)=\mathcal{F}^1(x)$, then see the processes as follows:
\begin{itemize}
  \item First step. Choose $x^1\in\mathbb{B}\setminus\mathbb{B}_0$, then find all possible
  $f(x^1)\in\{0,1\}^R$ satisfying $d_\a (x^1,\xp)\geq d_\b(f(x^1),f(\xp))$ for all $f(\xp)\in\mathcal{F}^2(\xp)$
  of all $\xp\in\mathbb{B}\setminus\{x^1\}$. Then replace $\mathcal{F}^2(x^1)$ with the set of
  all possible $f(x^1)$ we just get.
  \item Second step. Choose $x^2\in\mathbb{B}\setminus\mathbb{B}_0$, then find all
  possible $f(x^2)\in\{0,1\}^R$ satisfying $d_\a (x^2,\xp)\geq d_\b(f(x^2),f(\xp))$ for all
  $f(\xp)\in\mathcal{F}^2(\xp)$ of all $\xp\in\mathbb{B}\setminus\{x^2\}$. Then replace
  $\mathcal{F}^2(x^2)$ with the set of all possible $f(x^2)$ we just get.
  \item ......
  \item $m^{\text{th}}$ step. Choose $x^m\in\mathbb{B}\setminus\mathbb{B}_0$, then find all
  possible $f(x^m)\in\{0,1\}^R$ satisfying $d_\a (x^m,\xp)\geq d_\b(f(x^m),f(\xp))$ for all
  $f(\xp)\in\mathcal{F}^2(\xp)$ of all $\xp\in\mathbb{B}\setminus\{x^m\}$. Then replace
  $\mathcal{F}^2(x^m)$ with the set of all possible $f(x^m)$ we just get.
\end{itemize}
One can repeat this process to get $\mathcal{F}^n(x)$ for all $x\in\mathbb{B}$.
Since $\mathcal{F}^{n+1}(x)\subset\mathcal{F}^{n}(x)$, it means for $\forall x\in\mathbb{B}$, $\{\mathcal{F}^{n}(x)\}$ is a set sequence monotonically decreasing with $n$, where the number of elements of $\mathcal{F}^n(x)$ is finite. Then $\exists N, \forall n>N, \mathcal{F}^n(x)$ will no longer decrease as $n$ increases.
We denote $\lim_{n \to \infty} \mathcal{F}^n(x)$ as $\mathcal{F}(x)$. $\mathcal{F}(x)$ is
either a certain nonempty set or empty set.

\begin{figure}[H]
  \centering
  \subfigure[Step 0]{
  \includegraphics[width=0.4\textwidth]{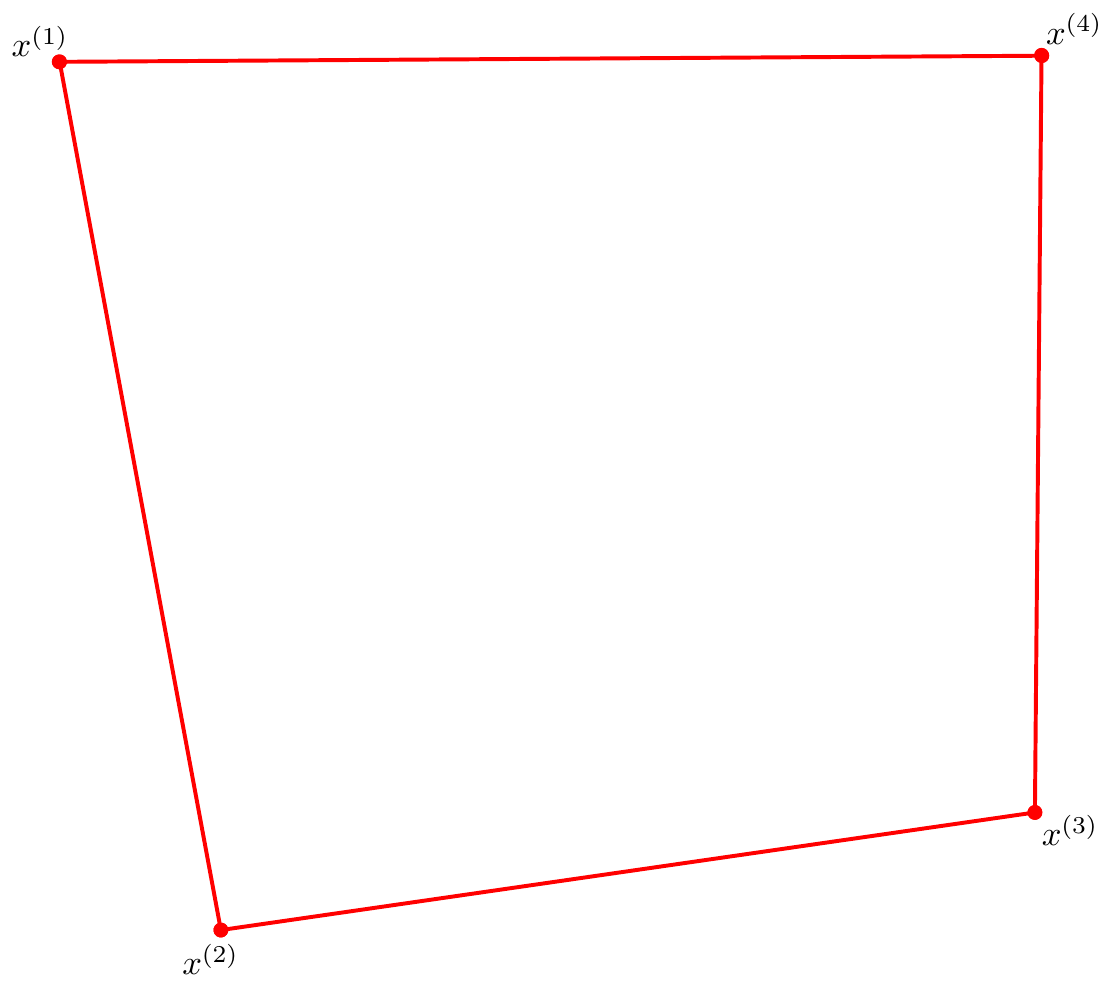}
  \label{MMI_graph_1}}
  \subfigure[Step 1]{
      \includegraphics[width=0.4\textwidth]{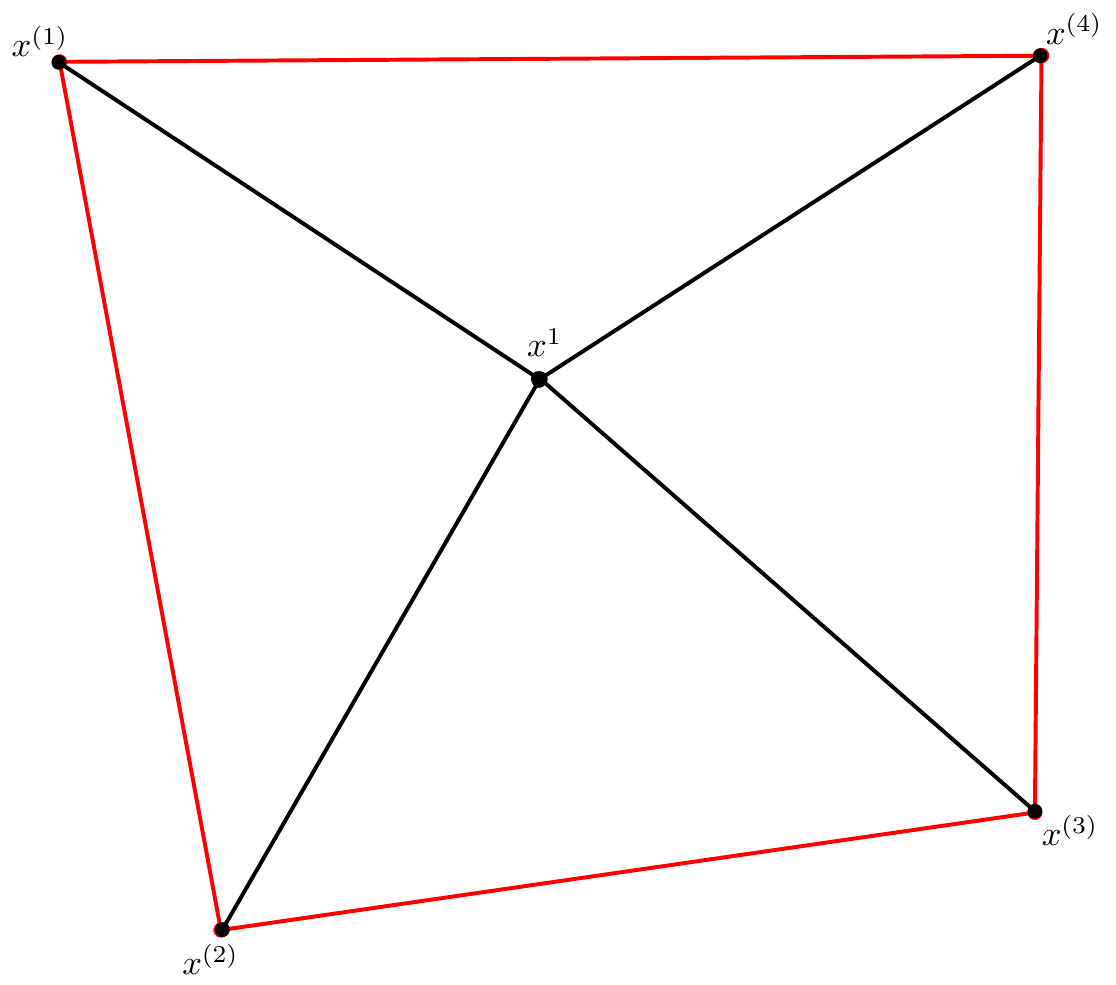}
      \label{MMI_graph_2}}\\
  \subfigure[Step 2]{
      \includegraphics[width=0.4\textwidth]{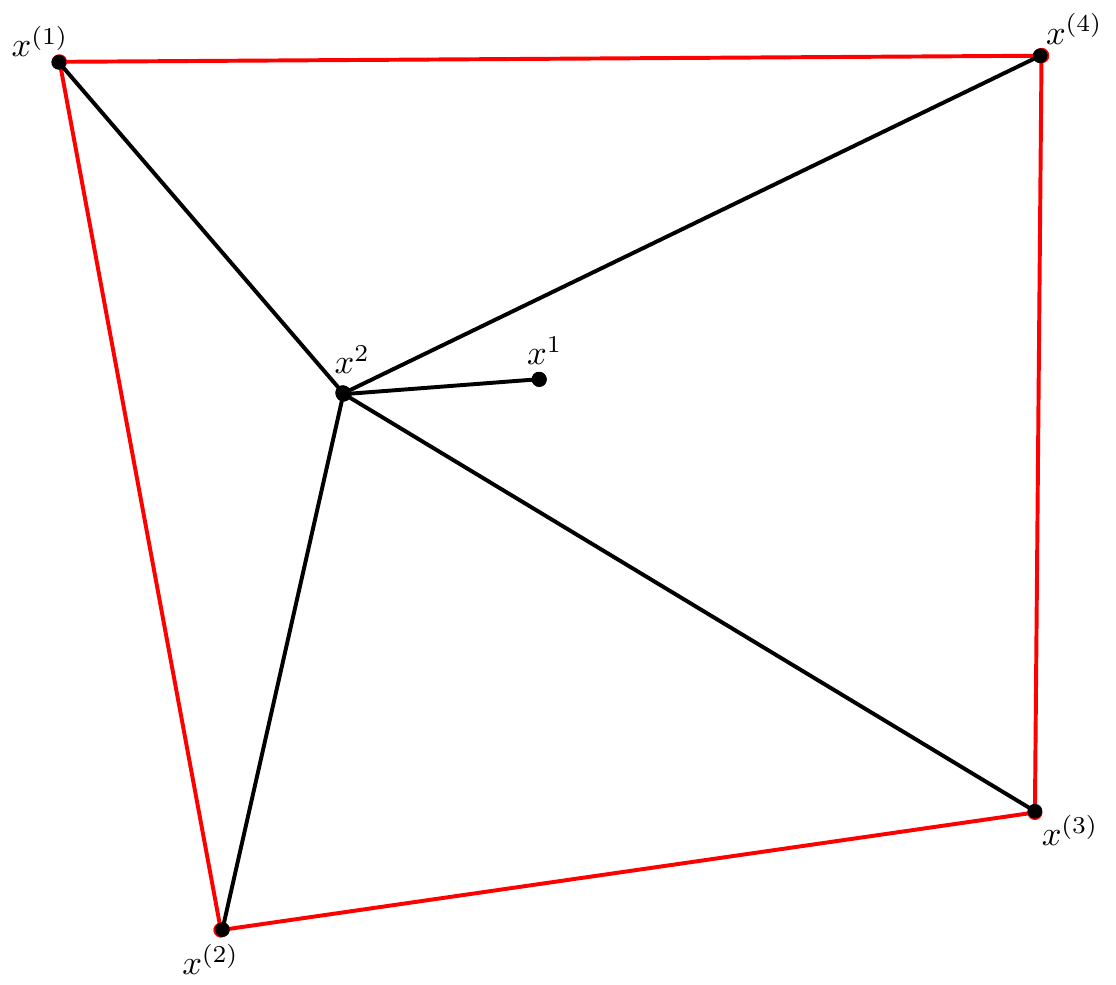}
      \label{MMI_graph_3}}
  \subfigure[Step m]{
      \includegraphics[width=0.4\textwidth]{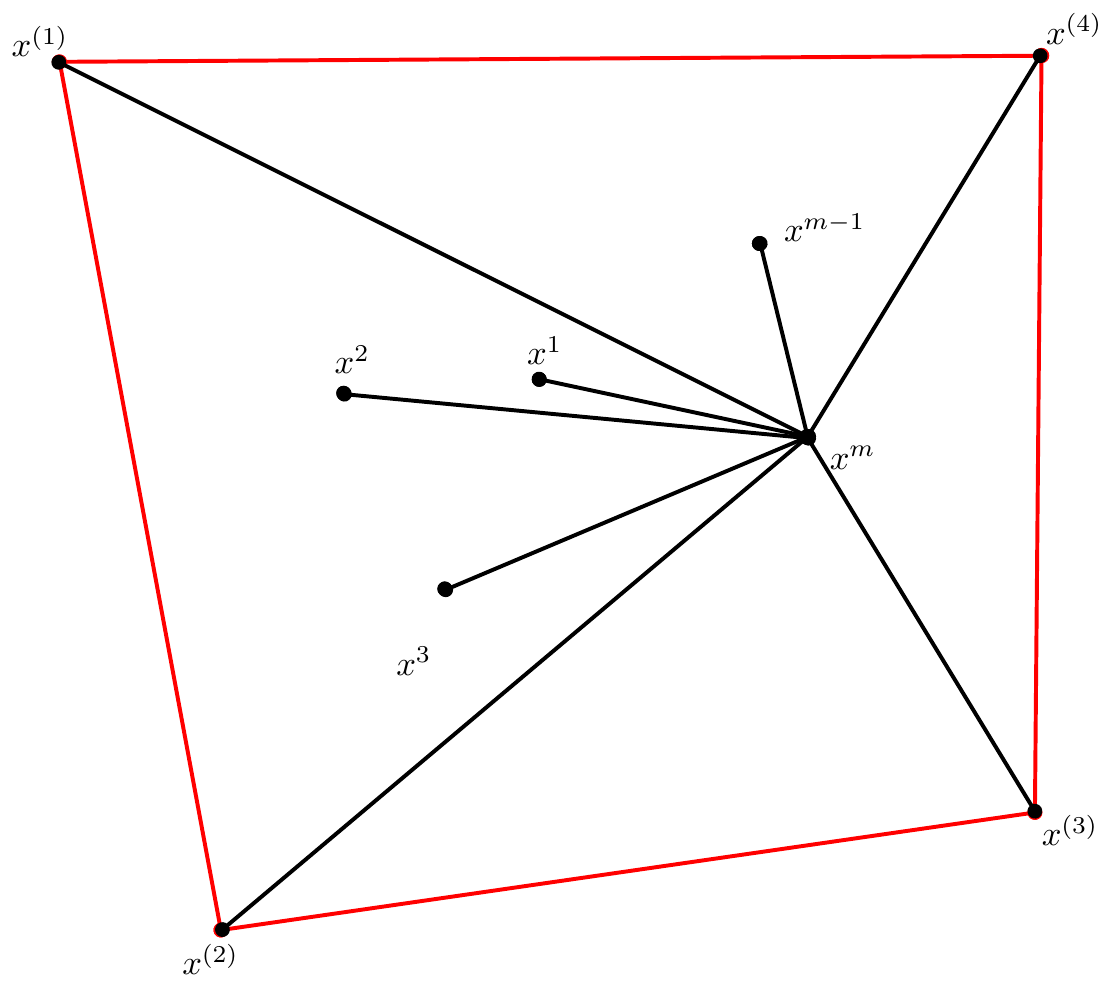}
      \label{MMI_graph_4}}
  \caption{There are fixed $f(x^{(i)})$ and uncertain $f(x^j)$.
  Step 1, in Fig(b), we add $x^1$ and connect it with all appeared point,
  each connection satisfying $d_\b\leq d_\a$(is short for $d_\b\fxx\leq d_\a\xx$).
  Every step is the same as step 1, we add a point and
  connect it with all appeared points, each connection satisfying $d_\b\leq d_\a$.}
  \label{MMI_graph}
\end{figure}

And if we first choose $x^j\in\mathbb{B}$ so that $d(x^j,x^{(i)})=1$, then the number of
possible $f(x^j)$ may be less than the choice of the random $x^j\in\mathbb{B}$. The time
complexity of the first part is of order $[N(\mathbb{B})]^2/2$ for the first process, while
$[N(\mathbb{B})]^2$ for the later process. If we repeat the process $a$ times in total, then the whole
time complexity of first part is of order $(2a-1)[N(\mathbb{B})]^2/2$.

Now let us turn to the second part. If for all $x\in\mathbb{B}$, we can find $f(x)\in\mathcal{F}(x)$
so that $f$ satisfies inequalities of weight distance constraints for all connected pairs $\xx$,
then this $f$ is what we want and thus we prove the inequality successfully. This is a global search
whose time complexity is of order $\prod_{x\in\mathbb{B}}N(\mathcal{F}(x))$
($N(\mathcal{F}(x))$ is the number of elements of $\mathcal{F}(x)$) which is huge.
However, it can imply whether the inequality is compact or not. If $\prod_{x\in\mathbb{B}}N(\mathcal{F}(x))$
is huge, then we might find a large number of possible $f$, which means the constraints of
weight distance inequality on $f$ could be tighter, i.e. the inequality that we prove might not be compact. So our method is still efficient for proving the compact inequality.

\subsection{Case with the overlapping relative homologous surface}\label{discuss}
In this subsection let us focus on a special case where some relative homologous surfaces
$\tm(I_l)$ and $\tm(I_{l^\prime})$  are overlapping, as depicted in
Fig.~\ref{examples about overlap a}. In this case, there is a caveat that the Lemma
\ref{area lemma} and Theorem \ref{nonzero area theorem} naively
seem to be no longer valid, since now the geometrically neighboring regions are no longer
directly connected as already mentioned in section \ref{4.2}. To rescue this case, our strategy is the
following: we first suppose there is a nonempty small region $\sigma(\epsilon)$ between
$\tm(I_l)$ and $\tm(I_{l^\prime})$ by taking an infinitesimal transformation on $\tm(I_l)$
and $\tm(I_{l^\prime})$ as shown in Fig.~\ref{examples about overlap b},
such that $\text{Area}[\tm(I_l)\cap\tm(I_l^\prime)]$ is equal to $0$, so all previous procedures and
theorems are valid. Then we can take a limit
$\sigma(\epsilon)\rightarrow \emptyset$. Because infinitesimal change of the area of $\tm(I_l)$
does not affect the method of finding $f$, the Lemma \ref{area lemma} and
Theorem \ref{nonzero area theorem} are robust under this limit.
Our method  of finding $f$ is still valid even for the case with the overlapping relative
homologous surface.

\begin{figure}[H]
  \centering
  \subfigure[]{
    \includegraphics[width=0.4\textwidth]{Directly_connected_1.pdf}
    \label{examples about overlap a}
    }
  \subfigure[]{
    \includegraphics[width=0.4\textwidth]{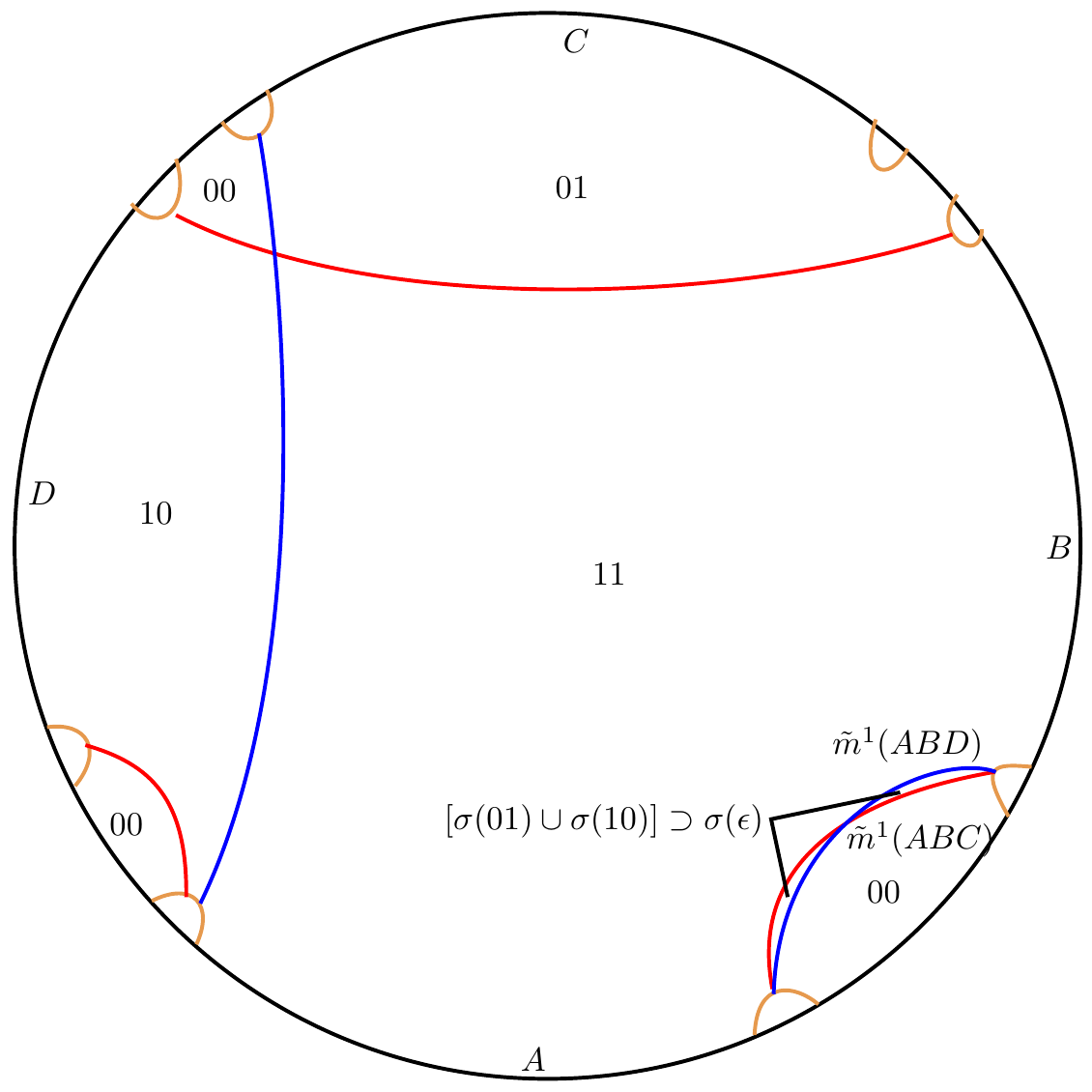}
    \label{examples about overlap b}
    }
  \caption{
  We take an infinitesimal transformation for $\tm(ABC)$ and $\tm(ABD)$ in figure (a),
  then we can get figure (b) so that $\sigma(\epsilon)\subset\sigma(01)\cup\sigma(10)$ and
  $\text{Area}[\tm(ABC)\cap\tm(ABD)]=0$.}
  \label{examples about overlap}
\end{figure}


\section{Examples}

By the definition of RHE, we expect the RHE possesses plentiful same inequalities as HEE. Related
work \cite{DSZ} also pointed out this from the ``flow" perspective. While we will assume some
inequalities of RHE and then prove these inequalities by using the above improved
proof-by-contraction method from the ``cut" perspective.

\subsection{Strong Subadditivity (SSA)}

We expect the SSA of RHE for a tripartite mixed state $\rho_{ABC}$, written as
\begin{equation}
    \ts_{AC}+\ts_{BC}\geq\ts_C+\ts_{ABC}.
\end{equation}
To prove this inequality, let us set
$A_1=A,A_2=B,A_3=C,A_4=O$,$I_1=\{1,3\},I_2=\{2,3\}$,$J_1=\{3\},J_2=\{1,2,3\}$,
$d_\a\xx=\sum_{l=1}^2 \abs*{x_l-x^{\prime}_l}$ and
$d_\b\fxx=\sum_{r=1}^2 \abs*{f(x)_r-f(x^{\prime})_r}$
Then basic constraints can be found in Table \ref{tab1}.
\begin{table}[H]
    \centering
    \begin{tabular}{l|cc}
        \toprule
    &$x^{(i)}$&$f(x^{(i)})$\\
    \hline
    1&10&01\\
    2&01&01\\
    3&11&11\\
    4&00&00\\
    \bottomrule
    \end{tabular}
    \caption{The basic constraints of $f(x)$ about SSA}\label{tab1}
\end{table}
Obviously all $x$ corresponds to a nonempty bulk piece, and all $f(x)$ are determined.
It is easy to verify that $d_\b\fxx\leq d_\a\xx$ for all directly connected $\xx$.
So the SSA of RHE is proved.


\subsection{Monogamy of mutual information (MMI)}\label{sect5.2}
We assume the MMI of RHE for a tripartite mixed state $\rho_{ABC}$, as
\begin{equation}
    \ts_{AB}+\ts_{AC}+\ts_{BC}\geq\ts_A+\ts_B+\ts_C+\ts_{ABC}.
\end{equation}

To prove this, we set:\\
$A_1=A,A_2=B,A_3=C,A_4=O$,\\
$I_1=\{1,2\},I_2=\{1,3\},I_3=\{2,3\}$,\\
$J_1=\{1\},J_2=\{2\},J_3=\{3\},J_4=\{1,2,3\}$,\\
$d_\a\xx=\sum_{l=1}^3 \abs*{x_l-x^{\prime}_l}$,
$d_\b\fxx=\sum_{r=1}^4 \abs*{f(x)_r-f(x^{\prime})_r}.$

Basic constraints can be found in Table \ref{tab2}.
\begin{table}[H]
    \centering
    \begin{tabular}{l|cc}
        \toprule
    &$x^{(i)}$&$f(x^{(i)})$\\
    \hline
    1&110&1001\\
    2&101&0101\\
    3&011&0011\\
    4&000&0000\\
    \bottomrule
    \end{tabular}
    \caption{The basic constraints of $f(x)$ about MMI}\label{tab2}
\end{table}
First, find all nonempty bulk pieces. We find other $x$ expect $x^{(i)}$ also corresponds to
nonempty bulk piece.We get
$\mathbb{B}\setminus\left(\bigcup_{i\in[4]}\{\ix\}\right)=\{100,010,001,111\}$.

Then, choose $x^1=100$, find all $f(x^1)$ satisfying $d_\b(\fix,f(x^1))\leq d_\a(\ix,x^1)$,
we get $\mathcal{F}(x^1)=\{0001\}$.\\
Choose $x^2=010$ find all $f(x^1)$ satisfying
$d_\b(\fix,f(x^2))\leq d_\a(\ix,x^2)$ and $d_\b(f(x),f(x_2))\leq d_\a(x,x^2),f(x)\in\mathcal{F}(x^1)$,
we get $\mathcal{F}(x^2)=\{0001\}$.\\
Choose $x^3=001$, find all $f(x^3)$ satisfying$d_\b(\fix,f(x^3))\leq d_\a(\ix,x^3)$ and
$d_\b(f(x),f(x^3))\leq d_\a(x,x^3),f(x)\in\bigcup_{j\in[2]}\mathcal{F}(x^j)$,
we get $\mathcal{F}(x^2)=\{0001\}$.\\
Choose $x^4=111$, find all $f(x^4)$ satisfying$d_\b(\fix,f(x^4))\leq d_\a(\ix,x^4)$ and
$d_\b(f(x),f(x^4))\leq d_\a(x,x^4),f(x)\in\bigcup_{j\in[3]}\mathcal{F}(x^j)$,
we get $\mathcal{F}(x^4)=\{0001\}$.
We have found all $f(x)$, and we can verify that $d_\b\leq d_\a$ for all directly connected
$\xx$, so we prove the MMI of RHE.

\subsection{Inequality of four regions}
The above two examples correspond to three regions. Actually, our approach is valid for more
regions. This is a concrete example for four regions. Before proceeding, let us recall that the inequality
of HEE was discovered by \cite{Bao:2015bfa} for a quadripartite
mixed state $\rho_{ABCD}$ as follows:
\begin{equation}
  I(A:B|C)+I(A:B|D)+I(C:D)\geq I(A:B)
\end{equation}
where $I(A:B)$ is the mutual information defined as $I(A:B)\coloneqq S_A+S_B-S_{AB}$ and $I(A:B|C)$ is the conditional mutual information defined as $I(A:B|C)\coloneqq S_{AC}+S_{BC}-S_C-S_{ABC}$.
Now we can transform the above formula of HEE into the following form in terms of RHE:
\begin{equation}
  \ts_{AB}+\ts_{AC}+\ts_{AD}+\ts_{BC}+\ts_{BD}
  \geq \ts_{A}+\ts_{B}+\ts_{CD}+\ts_{ABC}+\ts_{ABD}.
\end{equation}

To prove this, we set:\\
$A_1=A,A_2=B,A_3=C,A_4=D,A_5=O$,\\
$I_1=\{1,2\},I_2=\{1,3\},I_3=\{1,4\}, I_4=\{2,3\},I_5=\{2,4\}$,\\
$J_1=\{1\},J_2=\{2\},J_3=\{3,4\},J_4=\{1,2,3\},J_5=\{1,2,4\}$,\\
$d_\a\xx=\sum_{l=1}^5\abs*{x_l-x^{\prime}_l}$\\
$d_\b\fxx=\sum_{r=1}^5 \abs*{f(x)_r-f(x^{\prime})_r}.$

Basic constraints can be found in Table \ref{tab3}.
\begin{table}[H]
    \centering
    \begin{tabular}{l|cc}
        \toprule
    &$x^{(i)}$&$f(x^{(i)})$\\
    \hline
    1&11100&10011\\
    2&10011&01011\\
    3&01010&00110\\
    4&00101&00101\\
    5&00000&00000\\
    \bottomrule
    \end{tabular}
    \caption{The basic constraints of $f$ about the inequality of four regions}\label{tab3}
\end{table}
First, find all nonempty bulk pieces. We find other $x$ expect $x^{(i)}$ also corresponds to
nonempty bulk piece. Then for all $x^{j}\in\mathbb{B}\setminus\left(\bigcup_{i\in[5]}\{\ix\}\right)$,
we find possible $f(x^{j})$ by following the steps we've introduced in section \ref{find f}.
We program to calculate the result of $f(x)$ in Table \ref{tab4}.
\begin{table}[H]
  \centering
  \begin{tabular}{|c|c||c|c|}
  \hline
$x$&$f(x)$&$x$&$f(x)$\\
\hline
  11010&00010&00100&00001\\
  11000&10010&00011&00011\\
  10101&00111&00010&00010\\
  10100&00011&00001&00001\\
  10010&01010&\bf{00000}&\bf{00000}\\
  10001&00011&\bf{11100}&\bf{10011}\\
  10000&00010&\bf{10011}&\bf{01011}\\
  01100&00011&\bf{01010}&\bf{00110}\\
  01000&00010&\bf{00101}&\bf{00101}\\
\hline
  \end{tabular}
  \caption{$f(x)$ of all $x\in\mathbb{B}$, bold numbers are the basic constraints of $f$}\label{tab4}
\end{table}

We've already got all $f(x)$, and we can verify
that $d_\b\leq d_\a$ for all directly connected $\xx$, so we prove this inequality of RHE.


\subsection{Inequality of five regions}
An inequality of HEE that is discovered by \cite{HEC2}, then we expect the same inequality of RHE
for a 5-partite mixed state $\rho_{ABCDE}$ as follows:
\begin{equation}
  \ts_{ABC}+\ts_{ABD}+\ts_{ACE}+\ts_{BDE}+\ts_{CDE}
  \geq \ts_{AB}+\ts_{AC}+\ts_{BD}+\ts_{CE}+\ts_{DE}+\ts_{ABCDE}.
\end{equation}

To prove this, we set:\\
$A_1=A,A_2=B,A_3=C,A_4=D,A_5=E,A_6=O$,\\
$I_1=\{1,2,3\},I_2=\{1,2,4\},I_3=\{1,3,5\},I_4=\{2,4,5\},I_5=\{3,4,5\}$,\\
$J_1=\{1,2\},J_2=\{1,3\},J_3=\{2,4\},J_4=\{3,5\},J_5=\{4,5\},J_6=\{1,2,3,4,5\}$,\\
$d_\a\xx=\sum_{l=1}^5 \abs*{x_l-x^{\prime}_l}$,$d_\b\fxx=\sum_{r=1}^6 \abs*{f(x)_r-f(x^{\prime})_r}.$

Basic constraints can be found in Table \ref{tab5}.
\begin{table}[H]
    \centering
    \begin{tabular}{l|cc}
        \toprule
    &$x^{(i)}$&$f(x^{(i)})$\\
    \hline
    1&11100 &	110001\\
    2&11010  &	101001\\
    3&10101  &	010101\\
    4&01011  &	001011\\
    5&00111  & 	000111\\
    6&00000 &  000000\\
    \bottomrule
    \end{tabular}
    \caption{The basic constraints of $f$ about the inequality of five regions}\label{tab5}
\end{table}
First, find all nonempty bulk pieces. We find other $x$ expect $x^{(i)}$ also corresponds to
nonempty bulk piece. Then for all $x^{j}\in\mathbb{B}\setminus\left(\bigcup_{i\in[6]}\{\ix\}\right)$,
we find possible $f(x^{j})$ by following the steps we have introduced in section \ref{find f}.
We program to calculate the result of $f(x)$ in Table \ref{tab6}.

\begin{table}[H]
  \centering
  \begin{tabular}{|c|c||c|c||c|c||c|c|}
  \hline
$x$&$f(x)$&$x$&$f(x)$&$x$&$f(x)$&$x$&$f(x)$\\
  \hline
11111&000001&01111&000011&10111&000101&01000&000001\\
11110&100001&01110&000001&10110&000001&00110&000101\\
11101&010001&01101&000001&10100&010001&00101&000101\\
11011&001001&01100&100001&10011&000001&00100&000001\\
11001&000001&01010&001001&10010&100001&00011&000011\\
11000&100001&01001&001001&\bf{11100}&\bf{110001}&\bf{11010}&\bf{101001}\\
10000&000001&00001&000001&\bf{10101}&\bf{010101}&\bf{01011}&\bf{001011}\\
10001&010001&00010&000001&\bf{00111}&\bf{000111}&\bf{00000}&\bf{000000}\\
\hline
  \end{tabular}
  \caption{$f(x)$ of all $x\in\mathbb{B}$, bold numbers are the basic constraints of $f$}\label{tab6}
\end{table}
We've found $f$ for all $x\in\mathbb{B}$, we can verify
that $d_\b\leq d_\a$ for all directly connected $\xx$, so we prove the inequality.


\section{Conclusion and Discussion}

In this work, we systematically study the proof-by-contraction method, an important
technique developed for proving a given holographic entropy inequality so as to fix HEC
of a system.  This method includes a contraction map $f$, which plays a central role in this
method and is of particular difficulty to construct for more-region ($n \geq 4$) cases. Here
we develop, based on a pioneer work \cite{Bao:2015bfa,Akers:2021lms},  a general and effective rule to
rule out most of the cases such that $f$ can be obtained in a relatively simple way and can
be processed by computer programming. We achieve this by carefully investigating all possible
constraints imposed on $f$, the basic constraints and constraints of inequality of weight
distance, which is not fully developed in the original method. Compared with the original method proposed in  \cite{Bao:2015bfa}, our method has two obvious improvements.
Firstly, the time complexity of
first part of our algorithm is of order $[N(\mathbb{B})]^2$, which is less than the order $(2^L)^2$
of the algorithm in \cite{Bao:2015bfa}. The second part is a global search for $f$. Secondly, our algorithm is globally optimal rather
than locally optimal, implying that more possible $f$ can be found. And since we only consider
$x\in\mathbb{B}$ instead of $x\in\{0,1\}^L$, the constraints on $f$
is less than the greedy algorithm \cite{Bao:2015bfa}, thus again more possible $f$ can be found. So that our method can be used to prove some possible inequalities that the method in \cite{Bao:2015bfa} can not.
The validity of our method has been confirmed by several examples. Thus our method can be
used to confirm the authenticity of some unproved inequality. It also can shed light on finding
new inequalities of HEE.

On the other hand, we consider the concept of RHE in static slice by the notion of relative homology \cite{Headrick:2017ucz}, which is a generalization of holographic entanglement entropy that is suitable for characterizing the entanglement of mixed states.
The holographic entanglement entropy and entanglement wedge cross section
$E_W$ are two particular cases of RHE \cite{DCS,HHea}. In this paper, we extend the whole
framework of proof-by-contraction to RHE. We find that our method not only can be used to
prove the inequalities of HEE, but also can be used to prove the inequalities of $E_W$.

As a future direction, the completeness of our method is of great interest.  There are three aspects that will affect the completeness of the method. The first and the most important is, as mentioned in section \ref{Inequality of weight distance}, that the weight distance inequalities $d_\a(x,x^{\prime})\geq d_\b(f(x),f(x^{\prime}))$
for all bitstrings pairs encoded by the surface fragments with non-zero area contribution is simply a sufficient
condition for  inequality~\eqref{entropy inequality}. The necessity of the condition is still unproven partially because we cannot find, even for the MMI, the case where the area of a certain surface fragment becomes dominant among the area of all other surface fragments. Second, some
$x\in\mathbb{B}$ might correspond to empty bulk pieces, so there might exist unnecessary
constraints on $f$. Third, we require the weight distance inequality for $\xx:\forall x,\xp\in\mathbb{B}$,
so there also exist unnecessary constraints on $f$. We leave these for future work.

Another future direction worthy of mention is to generalize our method to prove the inequalities of RHE in the covariant case, as we focus on the static case in the paper. It was known that the SA, SSA and MMI of HEE can be proved to hold in the covariant case by Wall's ``maximin'' method with assumed null curvature condition \cite{EW2}, but this method encounters obstacles for the inequalities beyond the MMI \cite{Rota:2017ubr,HEC1}. So it is natural to think whether the RHE will encounter similar troubles. Unfortunately, as far as we know, the definition of the RHE in \cite{Headrick:2017ucz} was mainly set in the static case so far by using the max flow-min cut theorem, and the general covariant RHE remains to be studied. One may expect that the general covariant RHE can be realized with the help of analogous min flow-max cut theorem in the Lorentzian setting \cite{Headrick:2017ucz}, as done for the holographic complexity \cite{Pedraza:2021mkh,Pedraza:2021fgp}. However, the main difference between RHE and HEE is that, the RHE possesses a more general homology condition than HEE. We expect the RHE will encounter similar troubles like HEE in principle when we try to prove inequalities in the covariant case, as HEE itself can also be regarded as a special case of RHE. It remains to be further clarified in the future.

\acknowledgments
We would like to thank Wenhao Zhou for inspired discussion on finding contraction map $f$.
This work is partially supported by the National Natural Science Foundation of China with
Grant No. 11975116, and the Jiangxi Science Foundation for Distinguished Young Scientists
under Grant No. 20192BCB23007.

\end{document}